\documentclass[a4paper]{article}
\usepackage[utf8]{inputenc}
\usepackage{amssymb,amsthm,amsmath}
\usepackage{graphicx}
\usepackage{graphics}
\usepackage{latexsym}
\usepackage{mathptmx}
\usepackage{ifpdf}
\usepackage{epstopdf}
\usepackage{enumitem}
\usepackage{authblk}
\usepackage{color}

\newcommand{\bF}{\mathbb{F}}
\newcommand{\bG}{\mathbb{G}}
\newcommand{\bH}{\mathbb{H}}
\newcommand{\bP}{\mathbb{P}}
\newcommand{\bQ}{\mathbb{Q}}

\newcommand{\R}{\mathbb{R}}

\def\A{{\mathcal{A}}}
\def\B{{\mathcal{B}}}

\def\F{{\mathcal{F}}}
\def\G{{\mathcal{G}}}

\def\P{{\mathcal{P}}}

\newcommand{\beq}{\begin{equation}}
\newcommand{\eeq}{\end{equation}}
\newcommand{\bi}{\begin{itemize}}
\newcommand{\bd}{\begin{description}}
\newcommand{\ei}{\end{itemize}}
\newcommand{\ed}{\end{description}}

\newcommand{\bc}{\begin{center}}
\newcommand{\ec}{\end{center}}

\newcommand{\sg}{\sigma}

\newcommand{\Ind}[1]{{\bf{1}}_{\{{#1}\}}}

\newtheorem{thm}{Theorem}[section]
\newtheorem{cor}[thm]{Corollary}
\newtheorem{lem}[thm]{Lemma}
\newtheorem{prop}[thm]{Proposition}
\newtheorem{defn}[thm]{Definition}
\newtheorem{rem}[thm]{Remark}
\newtheorem{ex}[thm]{Example}
\newtheorem{ass}[thm]{Assumption}

\begin{document}

\title{Arbitrage
  and utility maximization in market models with an insider\footnote{
The research of Chau Ngoc Huy was supported by Natixis Foundation for
Quantitative Research {and the "Lend\"ulet" grant LP2015-6 of the Hungarian Academy of Sciences}. The research of Peter Tankov was supported by
the chair ``Financial Risks'' sponsored by Soci\'et\'e
G\'en\'erale.}}
\author[1]{Huy N. Chau}
\author[2]{Wolfgang J. Runggaldier}
\author[3]{Peter Tankov}
\affil[1]{Alfr\'ed R\'enyi Institute of Mathematics, Hungarian Academy of Sciences, Budapest}\affil[2]{University of Padova} \affil[3]{LPMA, Universit\'e
Paris-Diderot, corresponding author. E-mail: tankov@math.univ-paris-diderot.fr}

\date{}
\maketitle

\begin{abstract}
We study arbitrage opportunities, market viability and utility
maximization in market models with an insider. Assuming that an
economic agent possesses an additional information in the form of an
$\mathcal F_T$-measurable random variable $G$, we give criteria for
the No Unbounded Profits with Bounded Risk property to hold,
characterize optimal arbitrage strategies, and prove duality results
for the utility maximization problem faced by the insider. Examples
of markets satisfying NUPBR yet admitting arbitrage opportunities
are provided for both atomic and continuous random variables $G$.
\end{abstract}



\section{Introduction}
The aim of this paper is to study arbitrage opportunities and
utility maximization in market modes with an insider. 
Insider information is typically modeled by using the mathematical
theory of enlargement of filtration, where one distinguishes
initial, successive and progressive enlargement. In this paper we
restrict ourselves to the setting of initial enlargement by a random
variable $G$: at time zero the insider knows the realization of $G$,
which the ordinary agents only observe at the end of the trading
period that we shall assume to be finite. {Note that some concepts of
arbitrage under initial enlargement and progressive enlargement on an infinite horizon have recently
been studied in \cite{acciaio_arbitrage_2014} and
\cite{aksamit2015optional}}.  

Insider trading under initial enlargement of filtration has been
the object of interest of many papers, including but by no means
limited to 
\cite{danilova2010optimal,amendinger_monetary_2003,grorud_insider_1998,ankirchner_shannon_2006,ankirchner_initial_2011,amendinger_additional_1998,pikovsky_anticipative_1996,imkeller_free_2001}. The
majority of these papers work in a complete market setting and are concerned with the question of additional
utility of the insider; they find that when the variable $G$ is
$\mathcal F_T$-measurable and not purely atomic, this additional
utility is often infinite.

In contrast to these papers, our main interest lies in exploring
various concepts of arbitrage in the context of initial filtration
enlargement. In particular, we are interested in the following
questions.
\begin{itemize}
\item When does the market for the insider satisfy the property NUPBR
  (no unbounded profit with bounded risk)? The NUPBR condition, see \cite{karatzas_numeraire_2007}, or,
equivalently,  No Arbitrage
of the first kind (NA1) (see also NAA1 in \cite{kabanov1994large}, and BK in \cite{kabanov1997ftap}), boils down to assuming that no positive claim, which is not identically
zero, may be superhedged at zero price. It is the minimal
condition enabling one to solve portfolio optimization problems in a
meaningful way: in \cite{karatzas_numeraire_2007} it is shown that, without NUPBR, one has either no
solution or infinitely many.  NUPBR is robust with respect to
changes of numeraire, absolutely continuous measure change and, in
some cases,  change of reference filtration (see e.g.~\cite{fontana_weak_2013}).
Finally, it is also known \cite{karatzas_numeraire_2007} that NFLVR
(the classical assumption of no free lunch with vanishing risk) is
equivalent to NUPBR plus the classical no arbitrage assumption (NA), which
means that markets with NUPBR can still admit (unscalable) arbitrage
opportunities.
\item When does the market for the insider admit optimal arbitrage? We
  say that a financial market admits optimal arbitrage if there exists
  a strategy which allows to superhedge a unit amount with an initial
  cost which is strictly less than one in some states of nature (note
  that the initial cost of a strategy for the insider may be a random
  variable since the insider possesses a nontrivial information
  already at time $t=0$). We say that the optimal arbitrage is strong
  whenever the replication cost is strictly less than one with
  probability one. In other words, an optimal arbitrage
  strategy allows to replicate a risk-free zero coupon bond at a price
  which is strictly less than the initial price of this bond (see
  \cite{baldeaux2013liability} for possible uses of such strategies in
  the context of asset liability management for e.g., pension funds). 
\end{itemize} 

To address the above questions, we distinguish the
cases when the additional information is represented by a discrete
(atomic) random variable $G$ and when it is given by a random
variable $G$ which is not purely atomic. The discrete case is the
easier one, and allows us to provide full answers to the above questions. Namely, the
following results are shown to be true under natural assumptions in
the case when $G$ is discrete.
\begin{itemize}
\item The market for the insider satisfies the NUPBR property.
\item If the original market (for non-informed agents) is complete,
  then the market for the insider admits strong optimal arbitrage. If
  the original market is incomplete, the optimal arbitrage may or may
  not exist, and we give examples of both situations. 
\end{itemize}

The case when $G$ is not purely atomic is more difficult, and only
partial answers to the above questions are provided in this
paper for this case.  Our first contribution here is to establish a new necessary condition for the insider
market to satisfy the NUPBR property. This condition is, in
particular, violated by all complete markets, which means that
complete markets always admit an arbitrage of the first kind. In the
incomplete markets the situation is less clear, and we provide
examples of both an incomplete market violating NUPBR and of an
incomplete market for which NUPBR holds and logarithmic utility of the
insider is finite, although the market admits arbitrage opportunities.

In addition to the above results, we also address the problem of utility optimization for
the insider. In this context, our contribution is two-fold. First, we show that the utility maximization problem for the insider
  may be expressed in terms of the quantities (strategies, martingale
  measures) defined in the original market for the uninformed
  agents. This in turn allows us to develop an extension of the
  classical duality results for  utility maximization to market models with an insider. 
These results are first obtained in the case of a discrete initial
information $G$, and then extended to a non purely atomic $G$ with a
limiting procedure.

The rest of the paper is structured as follows.
In section \ref{S1} we introduce the market model and recall the basic
notions of no-arbitrage and filtration enlargement.

In section \ref{S2} we deal with
initial enlargement by a discrete random variable $G$. We show that, under a suitable assumption, a market
initially enlarged with a discrete r.v. $G$ always satisfies NUPBR.
In subsection \ref{S2.1} we then study optimal arbitrage that
can be implemented via superhedging. We show that the superhedging
price of a given claim for the insider may be represented in terms of the superhedging prices in the filtration of
ordinary agents of the claim restricted to the events
corresponding to the various possible values of $G$. In subsection
\ref{S2.2} we consider portfolio optimization, in particular
the maximization of expected utility and obtain a duality
relationship.  An example computation of optimal
arbitrage and maximal expected log-utility for the insider in an
incomplete market is presented in subsection \ref{ex:discrete_poisson}.

In section
\ref{section:G_continuous} we study the initial
enlargement with a random variable $G$, which is not purely atomic. We
first show that if the set of possible
martingale densities is uniformly integrable, then NUPBR cannot
hold. We then present an approximation procedure allowing to
obtain results for a general random variable $G$ by a limiting
procedure from the results obtained for a discrete variable $G$ in
section \ref{S2}. This procedure allows us to extend the results on
utility optimization to the case of general
$G$ in subsection \ref{section:NUPBR_continuous_G}.
Finally, the Appendix contains
some technical proofs.

\section{Market model and preliminaries related to filtration
enlargement}\label{S1} In this section we introduce our basic market
model and recall known concepts as we shall use them in the sequel
(subsection \ref{S1.1}). We then introduce some preliminaries in
relation to filtration enlargement (subsection \ref{S1.2}).

\subsection{Market model and basic notions}\label{S1.1}

On a stochastic basis $(\Omega,\F, \bF,\bP)$, where the filtration
$\bF=(\F_t)_{t\ge 0},\>t\le T$ satisfies the usual conditions,
consider a financial market with an $\R^d-$valued nonnegative
semimartingale process $S=(S^1,\dots,S^d),\>t\le T$, where the
components represent the prices of $d$ risky assets. The horizon is
supposed to be finite and given by $T>0$. We assume that the price
processes are already discounted, namely for the riskless asset
price $S^0$ we assume $S^0\equiv 1$, and that this market is
frictionless. Let $L(S)$ be the set of all $\R^d-$valued
$S-$integrable predictable processes and, for $H\in L(S)$ denote by
$H\cdot S$ the vector stochastic integral of $H$ with respect to
$S$. %In the following definitions we recall some known concepts as we shall use them in the sequel.
\begin{defn}\label{D1} An {\sl investment strategy $H$} is an element
$H\in L(S)$, where the components indicate the number of units
invested in the individual assets. Letting $x\in \R_{+}$, an $H\in
L(S)$ is said to be an $x-$admissible strategy, if $H_0=0$ and
$(H\cdot S)_t\ge -x$ for all $t\in [0,T], \bP-$a.s. $H\in L(S)$ is
said to be admissible if it is $x-$admissible for some $x\in
\R_{+}$. We denote by $\A_x$ the set of all $x-$admissible
strategies and by $\A$ that of all admissible strategies. For
$(x,H)\in\R_{+}\times \A$ we define the portfolio value process
$V_t^{x,H}:=x+(H\cdot S)_t$ implying that $x$ is the amount of the
initial wealth and that portfolios are generated only by
self-financing admissible strategies. Finally, we denote by $\mathcal
K_x$ the set of all claims that
one can realize by $x$-admissible strategies starting with zero
initial cost:
$$\mathcal{K}_x = \left\{ V^{0,H}_T \mid H \in \mathcal A_x\right\}$$
and $\mathcal K$ denotes the set of claims that can be replicated with
zero initial cost and any admissible strategy: $\mathcal K
= \cup_{x\geq 0} \mathcal K_x$.
\end{defn}
Let
$$
\mathcal C = (\mathcal K - L^0_+) \cap L^\infty.
$$
In the sequel, we shall use the following no-arbitrage conditions.
\begin{defn}[NFLVR]
We say that there is No Free Lunch with Vanishing Risk if
$$
\overline{\mathcal C} \cap L^\infty_+ = \{0\},
$$
where the closure is taken with respect to the topology of uniform
convergence.
\end{defn}
\begin{defn}[NUPBR]\label{D2}
There is No Unbounded Profit With Bounded Risk if the set $\mathcal{K}_1$
is bounded in $L^0$, that is, if
$$\lim_{c \uparrow \infty} \sup_{W \in \mathcal{K}_1} P(W>c) = 0$$
\end{defn}

The NUPBR condition can be shown to be equivalent to the following,
more economically meaningful condition, which boils down to assuming
that no positive claim, which is not identically
zero, may be superhedged at zero price with a positive portfolio (see \cite{kabanov2015no} for a
recent discussion of the different equivalent formulations of NUPBR).
\begin{defn}[NA1]\label{D3} An $\F_T-$measurable random variable $\xi$ is called
an Arbitrage of the First Kind if $\bP(\xi\ge 0)=1$, $\bP(\xi>
0)>0$, and for all $x>0$ there exists an admissible strategy $H\in
\A_x$ such that $V^{x,H}_T\ge \xi$. We shall say that the market
admits No Arbitrage of the First Kind (NA1), if no such random
variable exists.\end{defn}

\begin{defn}[Classical arbitrage]\label{D4}
We shall say that $H\in \mathcal A$ is an arbitrage strategy if
$P(V^{0,H}_T\geq 0) =1$ and $\bP(V^{0,H}_T> 0) >0$. It is a {\sl
strong arbitrage} if $\bP(V^{0,H}_T>0)=1.$ An arbitrage strategy is
said to be \emph{scalable}  if $H\in \A_0$ and \emph{unscalable} if
$H\in \A_x$ with $x>0$ for some $x$, but $H\notin \mathcal A_0$. We shall say that there is {\sl absence of classical
arbitrage}, denoted by NA, if there are no scalable or unscalable
arbitrage strategies, that is,
$$
\mathcal C \cap L^\infty_+ = \{0\}.
$$\end{defn}

In our context (nonnegative processes), NFLVR is equivalent to the existence of at least one equivalent
local martingale measure \cite{delbaen_general_1994}. The set of all
such measures will be denoted by $ELMM(\mathbb F,\mathbb P)$ and the
set of corresponding densities will be denoted by $ELMMD(\mathbb F,\mathbb P)$. NUPBR is, in
turn, equivalent to the existence of a local martingale deflator
\cite{kardaras_market_2012,takaoka_condition_2014,song_alternative_2013}. In
addition, NFLVR is equivalent to NUPBR plus NA
\cite{karatzas_numeraire_2007}.
% It follows from [DS] that, if NA fails, then the market can create
% either scalable or unscalable arbitrage. We shall be interested only
% in excluding unscalable arbitrage because it is the only one
% possible under NUPBR and it is also the only economically meaningful
% one. Recalling that (see....) NUPBR+NA=NFLVR, in the quest of models
% that are intermediate between NFLVR and NUPBR one may thus search
% for models, for which NUPBR holds, but (unscalable) NA does not
% hold. In the intermediate region between NFLVR and NUPBR there
% exists thus the possibility for classical (unscalable) arbitrage.
% These arbitrages may be exploited in various ways, one might thus be
% interested in finding optimal arbitrage, which will be the topic of
% subsections \ref{S2.1} and \ref{section:A1_G}.

\subsection{Preliminaries in relation to filtration
enlargement}\label{S1.2}

We start again from a filtered probability space $(\Omega,\F,
\bF,\bP)$, on which we consider a financial market with an insider.
\begin{ass}\label{basic-assum}
The $(\bF,\bP)$-market of regular agents satisfies NFLVR implying
that the set $ELMMD(\bF,\bP)$ is not empty. The insider possesses
from the beginning an additional information about the outcome of
some $\F_T-$measurable r.v. $G$ with values in $(\R,\B)$.\end{ass}
Starting from $\bF=(\F_t)$ one can then consider the (initially)
{\sl enlarged filtration} $\bG=(\G_t)$ with
$$\G_t=\displaystyle\cap_{\varepsilon>0}\,\left(\F_{t+\varepsilon}\vee \sg(G)\right)$$

In the context of filtration enlargements it is important to have a
criterion which ensures that an $\bF-$local martingale remains a
$\bG-$semimartingale. In view of introducing the corresponding
condition, let $\nu_t:=\bP\{G\in dx\mid\F_t\}$ be the regular
conditional distribution of $G$, given $\F_t$, and $\nu:=\bP\{G\in
dx\}$ be the law of $G$. We shall require Jacod's condition (see \cite{jacod_grossissement_1985})
in the following form

\begin{ass}\label{assum:jacod} {\sl (Absolutely
continuous version of Jacod's condition)}. We assume that
$$\nu_t\>\ll\>\nu,\quad \bP-a.s.\>\>\>\mbox{for}\>\>\>t<T.$$\end{ass}
Notice that the
absolute continuity is imposed only before the terminal time $T$. In
our setting, where $G\in \mathcal F_T$, the absolute continuity
cannot hold at the terminal date.
Some papers on insider trading require that $\nu_t\>\sim\>\nu$ (see,
e.g., \cite{amendinger_additional_1998}) but this would
imply that the density of $\nu_t$ with respect to $\nu$ is strictly
positive and so allow one to construct an equivalent martingale
measure from the density process \emph{before the terminal time $T$} \cite{amendinger_additional_1998}. This would imply NFLVR
for the $(\bG,\bP)-$ market (before time $T$) and thus exclude arbitrage possibilities
there, which is not our purpose.

We need one more assumption, which refers to the density process of
$\nu_t$ with respect to $\nu$. To this effect we first recall from
Lemme 1.8 and Corollaire 1.11 of \cite{jacod_grossissement_1985} that we can choose a nice
version of the density, namely we have the following lemma where
$\mathcal{O}(\mathbb{F})$ denotes the $\mathbb{F}$-optional sigma
field on $\Omega \times \mathbb{R}_+$.
\begin{lem}\label{lemma:density_hypo}
Under Assumption \ref{assum:jacod}, there exists a nonnegative
$\mathcal{B} \otimes \mathcal{O}(\mathbb{F}) $-measurable function
$\mathbb{R} \times \Omega \times \mathbb{R}_+ \ni (x, \omega, t)
\mapsto p^x_t(\omega) \in [0, \infty)$, c\`{a}dl\`{a}g in $t$ such
that
\begin{enumerate}
\item for every $t \in [0,T)$, we have $\nu_t(dx) = p^x_t(\omega)\nu(dx)$.
\item for each $x \in \mathbb{R}$, the process $(p^x_t(\omega))_{t \in [0, T)}$ is a $(\mathbb{F}, \mathbb{P})$-martingale.
\item The processes $p^x, p^x_-$ are strictly positive on $[0,\tau^x)$ and $p^x = 0$ on $[\tau^x, T)$, where
$$\tau^x := \inf \{t \ge 0: p^x_{t-} = 0 \text{ or } p^x_t = 0 \} \wedge T.$$
Furthermore, if we define $\tau^G(\omega) :=
\tau^{G(\omega)}(\omega) $ then $\mathbb{P}[\tau^G = T ] = 1.$
\end{enumerate}
\end{lem}
The conditional density process $p^G$ is also the key to find the
semimartingale decomposition of an $\bF-$local martingale in the
enlarged filtration $\bG$.

We come now to the announced additional assumption
\begin{ass}\label{assum:density_jump}
For every $x$, the process $p^x$ does not jump to zero, i.e.
$$\mathbb{P}[\tau^x < T, p^x_{\tau^x-}> 0] = 0.$$
\end{ass}
This assumption is used in \cite{kardaras_strict_2011} for a general construction
of strict local martingales, in \cite{ruf_systematic_2013} for a construction of markets
with arbitrages and in \cite{chau_market_2015} for the study of optimal arbitrage when
agents have non equivalent beliefs. {This assumption is also used to prove the preservation of NUPBR in the enlarged market over an infinite horizon, see  \cite{acciaio_arbitrage_2014} or Theorem 6(a) of \cite{aksamit2015optional}.}
%For a
%more intuitive meaning of this assumption see [C-thesis, ch 4].

\section{Enlargement with a discrete random variable}\label{S2}

In this section we consider the case when the random variable $G$ of
the (initial) enlargement is a discrete random variable
$G\in\{g_1,\ldots,g_n\}$, {with $n\geq 2$ and $\mathbb P[G = g_i]>0$
for all $i$.} After a general theorem concerning NUPBR
for this case, we study optimal arbitrage in subsection \ref{S2.1}
and provide a dual representation for expected utility maximization in
subsection \ref{S2.2}. An example for computing optimal arbitrage and
maximal expected utility for the insider is presented in
subsection \ref{ex:discrete_poisson} for the case of an
incomplete market.

Notice first that the initial enlargement with a discrete random
variable is a classical case studied already by P.~A.~Meyer
\cite{meyer1978theoreme} and by many other authors. For this case
it is known that every $\bF-$local martingale is a
$\bG-$semimartingale on $[0,T]$ and it is not necessary to impose
Jacod's condition. We shall however make the Assumption
\ref{assum:density_jump}.% for the various $x=g_i,\>i=1,\ldots,n$.

In the discrete case the insider can update her belief with a {\sl
measure change} $\bP\>\to\>\bQ^i$ thereby dismissing all scenarios not
contained in $\{G=g_i\}$. The measure $\bQ^i$ satisfies
\beq\label{Qi}\frac{d\bQ^i}{d\bP}{\big|_{\F_t}}=\frac{\bP\{G=g_i\mid\F_t\}}{\bP\{G=g_i\}}:=p_t^{g_i}\eeq
It gives total mass to $\{G=g_i\}$ and is absolutely continuous but
not equivalent to $P$.

The following theorem shows that NUPBR always holds true in this
setting.
\begin{thm}\label{T-NUPBR}
{Let $G$ be discrete and suppose that Assumptions \ref{basic-assum} and
\ref{assum:density_jump} hold true. Then} the $(\bG,\bP)-$market satisfies NUPBR.\end{thm}
\begin{proof}
The statement is proved by way of contradiction, noticing that NUPBR is
equivalent to NA1. Assume that there is an arbitrage of the first
kind in the $(\mathbb{G},\mathbb{P})$-market, i.e., we can find an
$\mathcal{F}_T$-measurable random variable $\xi$ (because $\mathcal{F}_T =
\mathcal{G}_T$) such that $\mathbb{P}[\xi \ge 0] = 1, \mathbb{P}[\xi
>0]>0$ and  for all $\varepsilon >0,$ there exists a
$\mathbb{G}$-predictable strategy $H^{\mathbb{G}, \varepsilon}$
which satisfies
\begin{equation}\label{eq:A1_G}
\varepsilon + (H^{\mathbb{G},\varepsilon} \cdot S)_T \ge \xi,\
\mathbb{P}-a.s.
\end{equation}
Choose an index $i$ such that $\mathbb{P[}\{\xi > 0\} \cap \{ G =
g_i\} ] > 0.$ The inequality (\ref{eq:A1_G}) still holds true under
$\mathbb{Q}^i$ in the form of
\begin{equation}\label{eq:A1_Gi}
\varepsilon + (H^{\mathbb{G}, \varepsilon}1_{G = g_i} \cdot S)_T \ge
\xi,\ \mathbb{Q}^i-a.s.
\end{equation}
Let us look at the hedging strategy $H^{\mathbb{G}, \varepsilon}1_{G
= g_i}$ under $\mathbb{Q}^i$. Recall that (see \cite{jeulin1980semi}) the predictable process $H^{\mathbb{G},
\varepsilon}$ is of the form $H_t^{\mathbb{G},
\varepsilon}(\omega)=h_t(\omega, G(\omega))$ where $h_t(\omega,x)$
is a $\P(\bF)\times\B(\R)-$measurable function with $\P(\bF)$
denoting the $\bF-$predictable $\sg-$algebra on
$\Omega\times\R_{+}$. Then
$H^{\mathbb{G},\varepsilon,i}:=h(\omega,g_i)$ is $\bF-$predictable
and we have the representation $H^{\mathbb{G},
\varepsilon}\Ind{G=g_i}=\tilde
H^{\mathbb{F},\varepsilon,i}\Ind{G=g_i}$ where
$\tilde{H}^{\mathbb{F},i, \varepsilon}$ is a
$\mathbb{F}$-predictable strategy. Thus $(\ref{eq:A1_Gi})$ implies
that $\xi$ is an arbitrage of the first kind in the $(\mathbb{F},
\mathbb{Q}^i)$-market, which is equivalent to the failure of NUPBR
in the $\mathbb{Q}^i$-market. Notice next that, for each $i \in
\{1,...,n\}$, the $(\mathbb{F}, \mathbb{Q}^i)$-market is obtained
from the $(\mathbb{F}, \mathbb{P})$-market by an absolutely
continuous measure change, see (\ref{Qi}). Furthermore, the density
process $p^{g_i}$ does not jump to zero, by Assumption
\ref{assum:density_jump}. By Theorem 4.1 of \cite{chau_market_2015} this implies that
the condition NUPBR holds for the $(\mathbb{F},
\mathbb{Q}^i)$-market thus proving the contradiction and with it the
statement.\end{proof}

Recently Acciaio et al.~\cite{acciaio_arbitrage_2014} gave sufficient conditions for NUPBR to
hold in the $(\bG,\bP)-$market by constructing a martingale deflator
under $G$ (see also the introductory part to section
\ref{section:G_continuous} below). {The construction of local martingale deflators is also given in Proposition 10 (for quasi left-continuous $\mathbb{F}$-local martingales), Proposition 11, and Theorem 6 of \cite{aksamit2015optional}. However they work on an infinite
horizon requiring the absolute continuity in Jacod’s hypothesis to hold at all times and so their approach cannot be adapted to our finite horizon case.} 

Theorem \ref{T-NUPBR} shows that, under the Assumption
\ref{assum:density_jump}, the $(\bG,\bP)-$market satisfies NUPBR; it
does not exclude that it satisfies also NFLVR. This depends on the
possibility of classical arbitrage in the various specific cases. The
study of such arbitrage opportunities is the subject of the next
section.

\subsection{Optimal arbitrage via superhedging}\label{S2.1}

The notion of optimal arbitrage goes back to
\cite{fernholz_optimal_2010}. Here, following
\cite{chau_market_2015}, we relate optimal arbitrage to
superhedging. We start from a definition of the superhedging price
which is adapted to the context of filtration enlargement. Whenever
in the sequel the filtration may be either $\bF$ of $\bG$, we shall
use the symbol $\bH\in\{\bF,\bG\}$.
\begin{defn}\label{DSH} Let $\mathbb{H} \in \{ \mathbb{F},
\mathbb{G} \}$ and let $f \ge 0$ be a given claim. An
$\mathcal{H}_0$-measurable random variable $x^\bH_*(f)$ is called the
superhedging price of $f$ with respect to $\mathbb{H}$ if there
exists an $\mathbb{H}$-predictable strategy $H$ such that
\beq\label{defi:superhedge_positive}\left\{\begin{array}{ll}
x^\bH_*(f) + (H \cdot S)_t &\ge 0, \qquad \mathbb{P}-a.s,  \forall t \in[0, T],  \\
x^\bH_*(f) + (H \cdot S)_T &\ge f, \qquad
\mathbb{P}-a.s.\end{array}\right.\eeq and, if any
$x\in\mathcal{H}_0$ satisfies these conditions, then $x^\bH_*(f)\le x,
\>\bP-$a.s.\end{defn}

In other words, the superhedging price of $f$
is the essential lower bound of the initial values of all nonnegative
admissible portfolio processes, whose terminal value dominates $f$.
Notice that
$\mathcal G_0$ is non trivial implying that the superhedging price
$x^\bG_*(f)$ is a random variable. However, this price is constant on
each event $\{G=g_i\}$.

We come next to the superhedging theorem that shows how the
superhedging price and a superhedging strategy for $f$ in $\bG$ can
be obtained in terms of the superhedging price and the associated
strategy in $\bF$ by restricting $f$ to the individual events
$\{G=g_i\}$.
\begin{thm}\label{T-SH} {Let $G$ be discrete and suppose that Assumptions \ref{basic-assum} and
\ref{assum:density_jump} hold true. Then},
\bi\item[i)] The {\sl
superhedging price} for a claim $f\ge 0$ in the $(\bG,\bP)-$market is
given by
$$x_{*}^{\bG}(f)=\sum_{i=1}^nx_{*}^{\bF}(f\Ind{G=g_i})\,\Ind{G=g_i}$$
\item[ii)]  {The associated {\sl hedging strategy} is $\sum_{i=1}^n H^{\bF,i}\Ind{G=g_i}$ where $H^{\bF,i}$ is the superhedging strategy
for $f\Ind{G=g_i}$ in the $(\bF,\bP)-$market, i.e.
$$\sum_{i=1}^n x_{*}^{\bF}(f\Ind{G=g_i})\,\Ind{G=g_i}+\sum_{i=1}^n\left(H^{\bF,i}\Ind{G=g_i}\cdot S\right)_T\ge f,\>\>
\bP-a.s.$$}\ei\end{thm}
\begin{rem}
{Using this theorem, the computation of the superhedging price for the
insider reduces to the computation of the superhedging price for the
uninformed agents in the original market, which satisfies NFLVR. In
particular, using the classical superhedging duality, we may write
$$x_{*}^{\bG}(f)=
\sum_{i}\sup_{Z\in ELMMD(\bF,\bP)}\mathbb E^\bP\left[Z f\Ind{G=g_i}\right]\,\Ind{G=g_i}. 
$$}
\end{rem}
\begin{proof}
As in the proof of Theorem \ref{T-NUPBR}, here we also make
use of Theorem 4.1 in \cite{chau_market_2015}, which relates the superhedging price under
a measure $\bP$ to that under a measure $\bQ$, with respect to which
$\bP$ is only absolutely continuous, but not necessarily equivalent.
The role of the measure $\bQ$ in \cite{chau_market_2015} will be played here by the
measure $\bP$ and that of $\bP$ in \cite{chau_market_2015} by the various measures
$\bQ^i$ defined in (\ref{Qi}). To make clear which measure is being used, in this proof
we shall use the notation $x_{*}^{\bF,\bP}(\cdot)$ or
$x_{*}^{\bF,\bQ^i}(\cdot)$ respectively.

Theorem 4.1 of \cite{chau_market_2015} leads to
$$x_*^{\mathbb{F}, \mathbb{Q}^i}(f) = x_*^{\mathbb{F}, \mathbb{P}}(f1_{G = g_i}).$$
For each $i$, we denote by $H^{\mathbb{F},i}$ the $\mathbb{F}$-predictable
strategy which superhedges $f$ in the
$(\mathbb{F},\mathbb{Q}^i)$-market, that is
$$x_*^{\mathbb{F}, \mathbb{Q}^i}(f) + (H^{\mathbb{F},i} \cdot S)_T \ge f,\quad \mathbb{Q}^i-a.s.$$
This inequality holds also under $\mathbb{P}$ when restricted on
$\{G=g_i\}$, namely
$$x_*^{\mathbb{F}, \mathbb{P}}(f1_{G = g_i})1_{G = g_i} +
(H^{\mathbb{F},i} 1_{G=g_i} \cdot S)_T \ge f1_{G = g_i}, \ \mathbb{P}-a.s.$$
Summing up these inequalities we obtain
$$ \left( \sum_i x_*^{\mathbb{F}, \mathbb{P}}(f1_{G = g_i})1_{G = g_i}\right)  +
\left(\left( \sum_i H^{\mathbb{F},i}1_{G = g_i} \right)\cdot S\right)_T \ge f,\
\mathbb{P}-a.s.$$ The hedging strategy $\left( \sum_i
H^{\mathbb{F},i}1_{G = g_i} \right)$ is $\mathbb{G}$-predictable.

Finally, we prove that the initial capital $\sum_i x_*^{\mathbb{F},
\mathbb{P}}(f1_{G = g_i})1_{G = g_i}$ is exactly the superhedging
price of $f$ in the $(\mathbb{G}, \mathbb{P})$-market. Assume that
$y$ is a $\mathcal{G}_0$-measurable random variable such that
$$ y + (H^{\mathbb{G}} \cdot S)_T \ge f, \ \mathbb{P}-a.s.$$ where  $H^{\mathbb{G}}$ is a $\mathbb{G}$-predictable strategy. Hence,
$$y1_{G=g_i} + (H^{\mathbb{G}}1_{G=g_i}\cdot S)_T \ge f1_{G = g_i},\ \mathbb{P}-a.s.$$
Because $\mathbb{Q}^i \ll \mathbb{P}$ with $\mathbb{Q}^i[G = g_i] =
1$, we obtain
$$y + (H^{\mathbb{G}}1_{G = g_i}\cdot S)_T \ge f, \ \mathbb{Q}^i-a.s.$$
By using the same argument as in the proof of Theorem \ref{T-NUPBR},
we can replace $H^{\mathbb{G}}$ on the various events $\{G=g_i\}$ by
an $\mathbb{F}$-predictable strategy $\tilde{H}^{\mathbb{F},i}$ and
then
$$y + (\tilde{H}^{\mathbb{F},i} \cdot S)_T \ge f, \ \mathbb{Q}^i-a.s.$$
By definition, the superhedging price of $f$ under $\mathbb{Q}^i$ is
not greater than $y$ and so we conclude that $\sum_i
x_*^{\mathbb{F}, \mathbb{P}}(f1_{G = g_i})1_{G = g_i} \le y,
\mathbb{P}-a.s.$
\end{proof}
We now give a specific definition of a market with optimal arbitrage which
is adapted to the context of initial filtration enlargement and is motivated by Lemma 3.3 in \cite{chau_market_2015}.
\begin{defn}\label{DOA}
There is {\sl optimal arbitrage} in the $(\bH,\bP)-$market if
$x^\bH_{*}(1)\le 1$ and $\bP\{x^\bH_{*}(1)<1\}>0$. If $x^\bH_{*}(1)<1, \bP-a.s.$
then the optimal arbitrage is said to be {\sl strong}.\end{defn}
{From this definition and Theorem \ref{T-SH} it follows that there is
optimal arbitrage for the insider if
and only if $x_{*}^{\bF}(\Ind{G=g_i})<1$ for some $i$. If the market
is complete, then $x_{*}^{\bF}(\Ind{G=g_i}) = \mathbb E^{\mathbb P}[Z
\Ind{G=g_i}]<1$ for all $i$, where $Z$ is the density of the unique
martingale measure. Therefore, in a complete market there always exists strong optimal arbitrage.}

It
follows from Remark 3.4 in \cite{chau_market_2015} that, under NUPBR, one has
$x^\bH_{*}(1)>0$ {$\bP$}-a.s. However, $x^\bH_{*}(1)>0$ does not imply NUPBR (see \cite{levental_necessary_1995}
for a market model that does not satisfy NUPBR, but satisfies NA,
which implies $x^\bH_{*}(1)>0$). %Recall also that $x_{*}^{\bG}(f)$ is
%random since $\G_0$ is non trivial, but $x_{*}^{\bG}(f)$ is constant
%on each event $\{G=g_i\}$.

% As a corollary to Theorem \ref{T-SH} we obtain immediately the
% following result, that we formulate as a theorem since it leads to
% optimal arbitrage which is the main purpose in this subsection.
% \begin{thm}\label{C-SH} If there is an index $i$ such that
% $x_{*}^\bF(\Ind{G=g_i})<1$, then the insider has {\sl optimal
% arbitrage on the event $\{G=g_i\}$}. If $x_{*}^\bG(1)<1$ then the
% insider has {\sl strong optimal arbitrage}.
% \end{thm}

% {\color{red} This is kind of trivial, should we really state this as a
%   theorem? }

%Notice that, if an arbitrage exists, NFLVR cannot hold.

\subsection{Expected utility maximization for an
insider}\label{S2.2}

This subsection concerns utility maximization. {We first
formulate the precise relationship between absence of arbitrage, in particular NUPBR, and
utility maximization, and then show that expected utility
maximization in the enlarged market can be performed by an analog of
{\sl classical duality} also under absence of an ELMMD.}

Given a concave and strictly increasing utility function $U(\cdot)$, the corresponding portfolio optimization problem is
given by
$$u(x):=\displaystyle\sup_{H\in{\cal A}_x}\mathbb E\left\{U(V_T^{x,H})\right\}$$
Recall that it is shown in \cite{karatzas_numeraire_2007} that, if
NUPBR fails, then $u(x)=+\infty$ for all $x>0$ or the problem has
infinitely many solutions. {This result implies immediately the
following statement, which we formulate as a proposition because of
its importance. It holds in general and not only
in the specific case of this section. }
\begin{prop}\label{L-NUPBR}Assume that the utility function is strictly
  increasing, concave, and satisfies
  $U(+\infty) = +\infty$. If there exists $x>0$ for which $u(x)<+\infty$, then NUPBR
holds.\end{prop} {In particular, a criterion allowing to show that NUPBR
holds is thus to show that e.g. the log-utility maximization leads
to a finite value. Notice, however, that NUPBR does not imply that
expected utility is finite.}

We now discuss the duality approach for utility maximization. 
Before stating the main theorem, we prove one preliminary lemma.
This lemma, which we state for a general increasing function,
shows that it is possible to relate the expected utility of the
insider to the expected utility of regular agents when restricted to
the events $\{G = g_i \}$. In this lemma and below we denote by
$\A_1^{\bF}$ and $\A_1^{\bG}$ the set of 1-admissible strategies
that are predictable with respect to $\bF$ and $\bG$ respectively.

\begin{lem}\label{lemma:log_utility_discrete}
{Let $U$ be an increasing function.} Then,
\begin{equation}\label{eq:log_utility_discrete}
\sup_{H \in \mathcal{A}^{\mathbb{G}}_1} \mathbb{E}^{\mathbb{P}}[U(
V^{1, H}_T)] = \sum_{i = 1}^n \sup_{H \in
\mathcal{A}^{\mathbb{F}}_1} \mathbb{E}^{\mathbb{P}}[U(V^{1,
H}_T)1_{G = g_i}].
\end{equation}\end{lem}
\begin{proof}
The proof only requires the representation of
$\mathbb{G}$-predictable processes as it was used in the proofs of
Theorems \ref{T-NUPBR} and \ref{T-SH}. We do not need Assumption
\ref{assum:density_jump} here. \smallskip

($\le$) Let $H^{\mathbb{G}} \in \mathcal{A}^{\mathbb{G}}$ be a
$\mathbb{G}$-predictable strategy. As mentioned above, the
$\mathbb{G}$-predictable process $H^{\mathbb{G}}$ can be expressed
as $H^{\mathbb{G}}_t(\omega) = h_t(\omega, G(\omega))$ where
$h_t(\omega, x)$ is a $\mathcal{P}(\mathbb{F}) \times
\mathcal{B}(\mathbb{R})$-measurable function. Then,
$H^{\mathbb{F},i} = h(\omega, g_i)$ is $\mathbb F$-predictable and
$H^{\mathbb{G}}1_{G = g_i} = H^{\mathbb{F},i}1_{G = g_i}$ a.s.
Hence, we have that $H^{\mathbb{G}} = \sum_{i=1}^n
H^{\mathbb{F},i}1_{G = g_i}$ and therefore
$$
\int_0^T H^{\mathbb G}_t dS_t = \sum_{i=1}^n 1_{G = g_i} \int_0^T
H^{\mathbb{F},i}_t dS_t,
$$
where the equality follows from the fact that $S$ is a $\mathbb
G$-semimartingale. Consequently
$$\mathbb{E}^{\mathbb{P}}[U(V^{1, H^{\mathbb{G}}}_T) ] = \sum_{i=1}^n \mathbb{E}^{\mathbb{P}}
\left[  1_{G = g^i} U\left(V_T^{1,H^{\mathbb{F},i}}\right)
\right],$$
{where we used the fact that $U$ is increasing to take the expectation
since it implies that both expressions under the $\mathbb E$ sign are
bounded from below.}
 Taking the supremum over the set of all
$\mathbb{G}$-admissible strategies we obtain the inequality $(\le)$
in (\ref{eq:log_utility_discrete}).\smallskip

($\ge$) {Let $H^{\bF,i}$, $i=1,\dots,n$ be $\bF$-predictable
strategies. Then, the strategy $H^{\bG}= \sum_{i=1}^n
H^{\mathbb{F},i}1_{G = g_i}$ is $\bG$-predictable and the following
straightforward inequality completes the proof.
$$\sum_{i=1}^n \mathbb{E}^{\mathbb{P}} \left[  1_{G = g^i} U(V_T^{1,H^{\mathbb{F},i}}) \right]
= \mathbb{E}^{\mathbb{P}}[U( V^{1, H^{\mathbb{G}}}_T) ] \le \sup_{H
\in \mathcal{A}^{\mathbb{G}}_1} \mathbb{E}^{\mathbb{P}}[U(V^{1,
H}_T)].$$}
\end{proof}

The following
theorem leads to a new characterization of the expected
utility of the insider in terms of the additional information
$G$ and the set of all local martingale densities of the
$(\mathbb{F},\mathbb{P})$-market.
\begin{thm} \label{thm:gen-dual}
Let $G$ be discrete, suppose that Assumptions \ref{basic-assum} and
\ref{assum:density_jump} hold true, and assume that 
 \begin{itemize}
\item[(i)] The function $U:(0,\infty)\to \mathbb R$ is strictly
  concave, increasing, continuously differentiable and satisfies the Inada conditions at
  $0$ and $\infty$.  
\item[(ii)] For every $y\in (0,\infty)$, there exists $Z
  \in 
  ELMMD(\bF,\bP)$ with $\mathbb E^{\mathbb P}[V(yZ)]<\infty$, where
  $V(y) =  \sup_x (U(x)-xy) $.
 \end{itemize}
Then, 
$$
\sup_{H\in \mathcal A^{\bG}_1} \mathbb E^{\bP} [U(V^{1,H}_T)] = \sum_{i}
\inf_{y>0}\left\{y+
\inf_{Z \in ELMMD(\mathbb{F},\mathbb{P})} \mathbb{E}^{\mathbb{P}}
\left[ V\left(yZ_T \right)1_{G = g_i}\right]\right\}.
$$
\end{thm}
\begin{proof}
In view of Lemma \ref{lemma:log_utility_discrete}, it suffices to show
that for every $i$, 
$$
\sup_{H\in \mathcal A^{\bF}_1} \mathbb E^{\bP} [\mathbf 1_{G=g_i} U(V^{1,H}_T)] = 
\inf_{y>0}\left\{y+
\inf_{Z \in ELMMD(\mathbb{F},\mathbb{P})} \mathbb{E}^{\mathbb{P}}
\left[ V\left(yZ_T \right)1_{G = g_i}\right]\right\}.
$$
Following the reasoning in the latter part of the proof of Theorem
\ref{T-NUPBR} that relates the $(\bF,\bQ^i)-$markets to the
$(\bF,\bP)-$market and using Assumption \ref{assum:density_jump},
one can again use Theorem 4.1 in \cite{chau_market_2015} to show that the $(\mathbb{F},
\mathbb{Q}^i)$-market satisfies the condition NUPBR. Furthermore,
for any local martingale density $Z \in
ELMMD(\mathbb{F},\mathbb{P})$, the process $Z/p^{g_i}$ is a local
martingale deflator for the $(\mathbb{F},\mathbb{Q}^i)$-market (note
that on $\{G = g_i\}$, $p_t^{g_i}>0$ for all $t$).

Let us next introduce the following subsets of $L^{0}_+$
%\begin{align*} \mathcal{C}(x) &= \{ v \in L^0_+: 0 \le v \le  xV^{1, H^{\mathbb{F}}}_T, \text{ for some } H^{\mathbb{F}} \in \mathcal{A}_1 \},\\ \mathcal{D}(y) &= \left\lbrace z \in L^0_+: 0 \le z \le yZ_T, \text{ for some } Z \in ELMMD(\mathbb{F}, \mathbb{P}) \right\rbrace,\\ \mathcal{D}^i(y) &= \left\lbrace z^i =  \frac{z}{p^{g_i}_T}, z \in \mathcal{D}(y) \right\rbrace.\end{align*}
{\color{red}
\begin{align*}
\mathcal{C}(x) &= \{ v \in L^0_+: 0 \le v \le x V^{1, H^{\mathbb{F}}}_T, \mathbb{P}-a.s. \text{ for some } H^{\mathbb{F}} \in \mathcal{A}^{\mathbb{F}}_1 \},\\
\mathcal{D}(y) &=\{ z \in L^0_+: 0 \le z \le yZ_T, \mathbb{P}-a.s. \text{ for some } Z \in ELMM(\mathbb{F}, \mathbb{P}) \},\\
\mathcal{C}^i(x) &= \left\lbrace v^i = v, \mathbb{Q}^i- a.s. \text{ for some } v  \in \mathcal{C}(x) \right\rbrace ,\\
\mathcal{D}^i(y) &=\left\lbrace z^i = \frac{z}{p^{g_i}_T}, \mathbb{Q}^i - a.s. \text{ for some } z \in \mathcal{D}(y) \right\rbrace .
\end{align*}
Because the $(\mathbb{F},\mathbb{P})$-market was assumed to satisfy
NFLVR, Proposition 3.1 of \cite{kramkov_asymptotic_1999} implies that $\mathcal{C}$ and
$\mathcal{D}$ are convex with the following properties
\begin{align*}
v \in \mathcal{C}(1) &\iff \mathbb{E}^{\mathbb{P}}[vz] \le 1, \text{ for all } z \in \mathcal{D}(1),\\
z \in \mathcal{D}(1) &\iff \mathbb{E}^{\mathbb{P}}[vz] \le 1, \text{
for all } v \in \mathcal{C}(1).
\end{align*}
These imply that for every $i$, 
%\begin{align*}v \in \mathcal{C}(1) &\iff \mathbb{E}^{\mathbb{Q}^i}[vz^i] \le 1, \text{ for all } z^i \in \mathcal{D}^i(1),\\z^i \in \mathcal{D}^i(1) &\iff \mathbb{E}^{\mathbb{Q}^i}[vz^i] \le 1, \text{ for all } v \in \mathcal{C}(1)\end{align*}
\begin{align*}
v^i \in \mathcal{C}^i(1)  &\Leftrightarrow \mathbb{E}^{\mathbb{Q}^i}[v^iz^i] \le 1, \text{for all } z^i \in \mathcal{D}^i(1),\\
z^i \in \mathcal{D}^i(1) &\Leftrightarrow \mathbb{E}^{\mathbb{Q}^i}[v^iz^i] \le 1, \text{for all } v^i \in \mathcal{C}^i(1).
\end{align*}
and thus the assumption (3.1) of \cite{mostovyi2015necessary} holds for
$\mathcal C^i(1)$ and $\mathcal D^i(1)$ under the measure
$\mathbb{Q}^i$. In addition, $\mathcal{C}^i$ and $\mathcal{D}^i$
contain at least one strictly positive element.

 Define, following
\cite{mostovyi2015necessary}, the optimization problems
$$
u(x) = \sup_{\xi \in \mathcal C^i(x)} \mathbb E^{\bQ^i}[U(\xi)]\quad
\text{and}\quad v(y) = \inf_{\eta\in \mathcal D^i(y)}\mathbb
E^{\bQ^i}[V(\eta)]. 
$$
For all $y >0$, the
finiteness of $v(y)$ follows from the assumptions of the Theorem. Furthermore, since
$x\in \mathcal C^i(x)$, we have $u(x)>-\infty$ for
all $x >0$. An
application of Theorem 3.2 of \cite{mostovyi2015necessary} then shows
that $u$ and $v$ satisfy biconjugacy relations so that in particular
$u(x) = \inf_{y>0}(v(y) + xy)$. Taking $x=1$ and substituting the
explicit expression for $\mathbb Q^i$, the proof is complete. 
}
\end{proof}

From  Theorem \ref{thm:gen-dual},  one immediately obtains more explicit
expressions for the case of power and logarithmic utility functions. 
\begin{cor} 
Fix $\gamma \in (0,1)$. Let $G$ be discrete, suppose that Assumptions \ref{basic-assum} and
\ref{assum:density_jump} hold true, and that there exists $Z\in
ELMMD(\bF,\bP)$ with $\mathbb E^{\bP} [(Z_T)^{-\frac{\gamma}{1-\gamma}}]<\infty$. Then, 
$$
\sup_{H\in \mathcal A^{\bG}_1} \mathbb E^{\bP} [(V^{1,H}_T)^\gamma] = \sum_{i}
 \left\{\inf_{Z  \in  ELMMD(\mathbb{F},\mathbb{P})}\mathbb{E}^{\mathbb{P}}
\left[ \left(Z_T \right)^{-\frac{\gamma}{1-\gamma}}1_{G = g_i}\right]\right\}^{1-\gamma}.
$$
\end{cor}

\begin{cor} \label{cor:log:dual}
Let $G$ be discrete, suppose that Assumptions \ref{basic-assum} and
\ref{assum:density_jump} hold true, and that there exists $Z \in 
ELMMD(\bF,\bP)$ with $\mathbb E^{\bP} [\log Z_T]>-\infty$. Then
\begin{align}
\sup_{H \in \mathcal{A}^{\mathbb{G}}_1} \mathbb{E}^{\mathbb{P}}[\log
V^{1, H}_T] = &- \sum_{i} \mathbb{P}[G= g_i]
\log \mathbb{P}[G = g_i], \nonumber  \\
&+  \sum_{i} \inf_{Z  \in  ELMMD(\mathbb{F},\mathbb{P})}
\mathbb{E}^{\mathbb{P}} \left[ 1_{G = g_i} \log
\frac{1}{Z_T}\right].\nonumber
\end{align}
\end{cor}

To conclude this subsection, we shall compare our result for the
logarithmic utility with the
results in \cite{amendinger_additional_1998}. Let the additional expected log-utility of the
insider be denoted by
$$\Delta(\mathbb{F},\mathbb{G}): =  \sup_{H \in \mathcal{A}^{\mathbb{G}}_1}
\mathbb{E}^{\mathbb{P}}[\log V^{1, H}_T] - \sup_{H \in
\mathcal{A}^{\mathbb{F}}_1} \mathbb{E}^{\mathbb{P}}[\log V^{1,
H}_T].$$ In the approach of \cite{amendinger_additional_1998}, the quantity
$\Delta(\mathbb{F},\mathbb{G})$ is represented by the information
drift, see Definition 3.6 in their paper, and in our approach, it
can be expressed as
\begin{align}\label{eq:log_diff_discrete}
- \sum_{i= 1}^n \mathbb{P}[G= g_i] \log \mathbb{P}[G = g_i]
\nonumber &+  \sum_{i = 1}^n
\inf_{Z \in ELMMD(\mathbb{F},\mathbb{P})} \mathbb{E}^{\mathbb{P}} \left[ 1_{G = g_i} \log \frac{1}{Z_T}\right]\\
&- \inf_{Z \in ELMMD(\mathbb{F},\mathbb{P})} \mathbb{E}^{\mathbb{P}}
\left[ \log \frac{1}{Z_T}\right].
\end{align}

If the market is complete, then the two approaches end up with the
same result: the quantity in (\ref{eq:log_diff_discrete}) reduces to
the entropy of $G$, which is exactly what is stated in Theorem 4.1
of \cite{amendinger_additional_1998}.

\subsection{Example}\label{ex:discrete_poisson}
In this subsection we present an explicit example which corresponds to
 an incomplete market. A complete market example based on
the standard Brownian motion is given in \cite{chau_thesis}.

Suppose that $N^1$ and $N^2$ are two independent Poisson
processes with common intensity $\lambda=1$. We consider a financial market with the (discounted)
risky asset price $S_t = e^{N^1_t - N^2_t}$ whose dynamics is given
by
$$dS_t = S_{t-}\left( (e-1) dN^1_t + (e^{-1} - 1) dN^2_t\right), \qquad S_0 = 1, \qquad t \in [0,T].$$
The public information $\mathbb{F}$ is generated by the two Poisson
processes  $N^1, N^2$. The $(\mathbb{F}, \mathbb{P})$-market
satisfies the NFLVR condition, and the density $Z$ of any equivalent
local martingale measure is of
the form
\begin{equation}\label{ex:poisson_density_F}
dZ_t = Z_{t-} \left( (\alpha^1_t - 1)  (dN^1_t - dt) +  (\alpha^2_t - 1)  (dN^2_t - dt) \right),
\end{equation}
where $\alpha^1, \alpha^2$ are positive integrable processes
satisfying $\alpha^1_t =
e^{-1}\alpha^2_t$.

Let us define $N_t := N^1_t - N^2_t$ and assume that the insider
knows the value of $N_T$, and hence of $S_T$, at the beginning
of trading. The insider's filtration is thus $\mathcal{G}_t =
\mathcal{F}_t \vee \sigma(N_T)=  \mathcal{F}_t \vee \sigma(S_T)$.
%Because the random variable $N_T$ takes values in $\mathbb{Z}$,
%Jacod's condition is satisfied. 
An easy computation shows that for
all $t \in [0,T)$,
\begin{equation*}
p^x_t = \frac{\mathbb{P}[N_T = x |\mathcal{F}_t]}{\mathbb{P}[N_T = x]} =
\frac{ \sum_{k \ge 0} e^{-(T-t)} \frac{(T-t)^k}{k!} e^{-(T-t)} \frac{(T-t)^{k+x-N_t}}{(k+x-N_t)!} 1_{k + x - N_t
\ge 0}  }{ \sum_{k \ge 0} e^{-T} \frac{T^k}{k!} e^{-T} \frac{T^{x+k}}{(x + k)!}  } > 0
\end{equation*}
and for $t = T$
$$ p^x_T = \frac{1_{N_T = x}}{ \sum_{k \ge 0} e^{-T} \frac{T^k}{k!}
  e^{-T} \frac{T^{x+k}}{(x + k)!}}.$$ Since the density $p_t^x$ is
strictly positive before time $T$,  Assumption (\ref{assum:density_jump}) is
fulfilled. Theorem
\ref{T-NUPBR} allows then to conclude that the
$(\mathbb{G},\mathbb{P})$-market satisfies NUPBR.

\subsubsection{Optimal arbitrage.} By Theorem
\ref{T-SH}, the superhedging price of $1$ under $\mathbb{G}$ is
$$x_*^{\mathbb{G},\mathbb{P}}(1) = \sum_{x \in \mathbb{Z}} x_*^{\mathbb{F}, \mathbb{P}}(1_{N_T = x})1_{N_T = x}.$$
To check whether optimal arbitrage exists, we need to compute $$x_*^{\mathbb{F},
\mathbb{P}}(1_{N_T = x}) = \sup_{\overline{\mathbb{P}} \in
ELMM(\mathbb{F},\mathbb{P})} \overline{\mathbb{P}}[ N_T = x]$$ for
every $x \in \mathbb{Z}$. This is the goal of the following
proposition.

\begin{prop}
If $x \le 0$, we have $$x_*^{\mathbb{F}, \mathbb{P}}(1_{N_T = x}) =
\sup_{\overline{\mathbb{P}} \in ELMM(\mathbb{F},\mathbb{P})}
\overline{\mathbb{P}}[ N_T = x] = 1$$ and there is no optimal arbitrage. If
$x > 0$, we have that $$x_*^{\mathbb{F}, \mathbb{P}}(1_{N_T = x}) =
\sup_{\overline{\mathbb{P}} \in ELMM(\mathbb{F},\mathbb{P})}
\overline{\mathbb{P}}[ N_T = x] = \frac{1}{e^x}$$ and the optimal
arbitrage strategy is the strategy in which the agent buys $\frac{1}{S_T}$ units
of the risky asset and holds them until maturity.
\end{prop}
\begin{proof}
First, we consider the case $x\le 0$. Let us define $\tau = \inf \{
t: N_t = x \}$. We fix two constants $M>m>0$ and choose  $\alpha^1_t
= M\mathbf 1_{t \le \tau} + m \mathbf 1_{t>\tau}$. The process
$\alpha^1_t$, and consequently also $\alpha^2_t$, is thus bounded
implying that for the density $Z$ one has
$\mathbb{E}^{\mathbb{P}}[Z_T]=1$ (see \cite[Theorem VI.T4]{bremaud_point_1981}).
Denote by $\overline{\mathbb{P}}^{M,m}$ the corresponding martingale
measure.
%This choice of $\alpha^1$ makes the Poisson processes $N^1$ and
%$N^2$ jump more frequently. However, when $N^1_t - N^2_t = x$, the
%two Poisson processes will not jump anymore. In other words, the
%measure $\overline{\mathbb{P}}^m$ concentrates on the event $N_T =
%x$.
For this measure, the following inequality holds true.
$$\sup_{\overline{\mathbb{P}} \in ELMM(\mathbb{F},\mathbb{P})}
\overline{\mathbb{P}}[ N_T = x] \ge
\mathbb{E}^{\overline{\mathbb{P}}^{M,m}}[1_{N_T = x}]  \geq
\overline{\mathbb{P}}^{M,m}[\tau \le T; N^1_t = N^1_\tau, N^2_t =
N^2_\tau\ \forall t\in [\tau,T]]. $$ By the strong Markov property,
$N^1_{\tau+s}-N^1_\tau$ and $N^{2}_{\tau+s}-N^2_\tau$ are
independent Poisson processes with intensities $m$ and $em$,
independent from $\mathcal F_\tau$. Therefore,
\begin{align*}
\overline{\mathbb{P}}^{M,m}[\tau \le T; N^1_t = N^1_\tau, N^2_t =
N^2_\tau\ \forall t\in [\tau,T]]&\geq
\overline{\mathbb{P}}^{M,m}[\tau \le T] \overline{\mathbb{P}}^{m}[
N^1_t = 0, N^2_t =0\ \forall t\in [0,T]] \\ &=
\overline{\mathbb{P}}^{M,m}[\tau \le T] e^{-m(1+e)T},
\end{align*}
where $\overline{\mathbb P}^m$ is the probability measure under which $N^1$ is a
Poisson process with intensity $m$ and $N^2$ is a Poisson process with
intensity $em$. On the other hand, up to time $\tau$, $N^1$ and $N^2$
are independent Poisson processes with intensities $M$ and
$eM$. Therefore,
$$
\overline{\mathbb{P}}^{M,m}[\tau\le T] = \overline{\mathbb{P}}^{M}\left[ \inf_{0\le t
    \le T}(N^1_t - {N}^2_t) \le x \right] = \overline{\mathbb{P}}^{1}\left[ \inf_{0\le t
    \le MT}(N^1_t - {N}^2_t) \le x \right]
$$
Letting $m$ go to zero and $M$ go to infinity and using the dominated convergence theorem, we
obtain
$$\sup_{\overline{\mathbb{P}} \in ELMM(\mathbb{F},\mathbb{P})} \overline{\mathbb{P}}[ N_T = x]
\ge \overline{\mathbb{P}}^1\left[ \inf_{ t \ge 0}(N^1_t - {N}^2_t) \le x \right] = 1,$$
because $N^1_t - {N}^2_t \to - \infty$ under $\overline{\mathbb P}^1$ as $t \to \infty$. So,
the first statement holds true.

Coming next to the case $x >0$, we notice that $e^{-x}$ is an upper
bound for the supremum. Indeed, for any ELMM
$\overline{\mathbb{P}}$, it holds that
$$ \overline{\mathbb{P}}[ N_T = x] \le \overline{\mathbb{P}}[ S_T \ge e^x] \le
\frac{\mathbb{E}^{\overline{\mathbb{P}}}[S_T]}{e^x} \le
\frac{1}{e^x}.$$ Repeating the computations as in the first case, we
obtain
\begin{multline*}\sup_{\overline{\mathbb{P}} \in ELMM(\mathbb{F},\mathbb{P})}
\overline{\mathbb{P}}[ N_T = x] \ge
\mathbb{E}^{\overline{\mathbb{P}}^{M,m}}[1_{N_T = x}] =
\overline{\mathbb{P}}^{M}[\tau \le T] e^{-m(1+e)T} \\=
\overline{\mathbb{P}}^1\left[ \sup_{0\le t \le MT}(N^1_t - {N}^2_t)
\ge x \right]e^{-m(1+e)T}.
\end{multline*} It thus suffices
to show that
$$f(x):=\overline{\mathbb{P}}^1\left[ \sup_{ t \ge 0}(N^1_t - {N}^2_t) \ge x \right] = \frac{1}{e^x}.$$
Let $\tau_1, \tau_2$ be the first jump times of $N^1$ and ${N}^2$,
respectively. Because $\tau_1 \sim Exp(1)$ and $\tau_2 \sim Exp(e)$
are independent under $\overline{\mathbb{P}}^1$, the random variable
$\frac{\tau_1}{e \tau_2}$ has the density $\frac{1}{(1+t)^2}$,
thanks to Lemma \ref{lem:ratio_dis_exp}, and thus,
$$\overline{\mathbb{P}}^1[\tau_1 < \tau_2]= \overline{\mathbb{P}}^1\left[ \frac{\tau_1}{e\tau_2} < \frac{1}{e} \right]
= \int\limits_0^{1/e} {\frac{1}{(1+t)^2}dt} =   \frac{1}{1+e}.$$
From its definition, we have $f(0)=1$ and for $x \ge 1$ it then
follows that
\begin{align*}
f(x) &= \overline{\mathbb{P}}^1\left[\sup_{ t \ge 0}(N^1_t - {N}^2_t) \ge x |\tau_1 > \tau_2 \right] \mathbb{P}[\tau_1 > \tau_2]\\
&+  \overline{\mathbb{P}}^1\left[\sup_{ t \ge 0}(N^1_t - {N}^2_t) \ge x |\tau_1 \le \tau_2 \right] \mathbb{P}[\tau_1 \le \tau_2]\\
&= \frac{ef(x+1)}{1+e}+\frac{f(x-1)}{1+e}.
\end{align*}
Therefore, we obtain $f(x+1) - f(x) = \frac{f(x) - f(x-1)}{e}$ and
thus
$$f(x) = 1-(1-f(1)) \frac{1-e^{-x}}{1 - e^{-1}}.$$
Because $\lim_{x \to \infty} f(x) = 0$, we have that $f(1) = e^{-1}$
and then $f(x) = e^{-x}$.

Now we show that the buy and hold strategy is optimal. Because the
insider knows the value of $S_T=e^{N_T}$, the buy and hold strategy,
consisting of $\frac{1}{S_T}$ units of the risky asset,
superreplicates the claim $1$. In fact
$$\frac{1}{S_T} + \frac{1}{S_T}\int\limits_0^T {1dS_u} = 1.$$
For this the insider needs the initial capital $e^{-x}$ on the event
$N_T = x$.
\end{proof}

{
\subsubsection{Expected log-utility.}\label{S.3.4.2} By Corollary
\ref{cor:log:dual} the expected log-utility of the insider is
\begin{align*}
\sup_{H \in \mathcal{A}^{\mathbb{G}}_1} \mathbb{E}^{\mathbb{P}}[\log V^{1, H}_T]
= - \sum_{ x \in \mathbb{Z}} \mathbb{P}[N_T = x] \log \mathbb{P}[N_T = x]
-  \sum_{x \in \mathbb{Z}} \sup_{Z \in ELMMD(\mathbb{F},\mathbb{P})} \mathbb{E}^{\mathbb{P}} \left[ 1_{N_T = x} \log Z_T\right].
\end{align*}
Choosing a specific strategy $H$ in the left-hand side, we obtain a lower bound
for the log utility, and a specific equivalent martingale measure
density $Z$ in
the right-hand side provides an upper bound. We are going to construct
$H$ and $Z$ for which the two bounds coincide. As a preliminary step, we are going to evaluate the
intensities of $N^1$ and $N^2$ under $\mathbb G$.
\paragraph{Intensities of $N^1$ and $N^2$ under $\mathbb G$}
Let $\lambda^{\mathbb{G},1}, \lambda^{\mathbb{G},2}$ be the
intensities of $N^1, N^2$ under $\mathbb{G}$, respectively. Introduce a
further larger filtration $\mathcal{H}_t = \mathcal{F}_t \vee
\sigma(N^1_T, N^2_T)$. Under $\mathbb{H}$, we obtain that
\begin{align}\label{ex:density_poisson_H}
dN^1_t - \frac{N^1_T - N^1_t}{T-t}dt, \qquad  dN^2_t - \frac{N^2_T - N^2_t}{T-t}dt
\end{align}
are martingales, see Theorem VI.3 of \cite{protter_stochastic_2003}. Now,
Lemma \ref{lem:compesator_small_filtration} of the Appendix implies
that the processes
\begin{align}\label{ex:density_poisson_G}
N^1_t - \int_0^t\mathbb{E}\left[ \left.   \frac{N^1_T - N^1_s}{T-s} \right| \mathcal{G}_s \right] ds,
\qquad  N^2_t - \int_0^t\mathbb{E}\left[ \left.  \frac{N^2_T - N^2_s}{T-s} \right| \mathcal{G}_s \right]ds
\end{align}
are martingales under $\mathbb{G}$ and so
\begin{align}
&\lambda^{\mathbb{G},1}_t = \mathbb{E}\left[ \left. \frac{N^1_T -
N^1_t}{T-t} \right| \mathcal{G}_t \right]\label{intenslambda}\\
&=\frac{1}{T-t}\,\mathbb{E}\left[N^1_{T} - N^1_{t}\mid N^1_{T} - N^1_t-
  N^2_{T} + N^2_t\right]=\frac{1}{T-t}f^1_t(N^1_{T} - N^1_t-
  N^2_{T} + N^2_t),\notag
\end{align}
where 
\begin{align*}
f^1_t(y) = \mathbb{E}[N^1_{T-t}| N^1_{T-t} - N^2_{T-t} = y]= \frac{
\mathbb{E}[N^1_{T-t}1_{N^1_{T-t} - N^2_{T-t} =
y}]}{\mathbb{P}[N^1_{T-t} - N^2_{T-t} = y]}.
\end{align*}
This computation can be made explicit. For example, if $y > 0$,  
\begin{align*}
f^1_t(y) &= \frac{ \sum_{k \ge 0} (y+k) \mathbb{P}[N^2_{T-t} = k] \mathbb{P}[N^1_{T-t} = y +k]  }
{\sum_{k \ge 0}\mathbb{P}[N^2_{T-t} = k] \mathbb{P}[N^1_{T-t} = y + k]}\\
&= \frac{ \sum_{k \ge 0}  (y + k )\frac{(T-t)^{2k+y}}{k!(y+k)!}  } { \sum_{k \ge 0} \frac{(T-t)^{2k +y}}{k! (k+y)!}}
= \frac{(T-t)I_{y-1}(2(T-t))}{I_y(2(T-t))},
\end{align*}
where $I_{\alpha}(x)$ is the modified Bessel functions of the first kind\footnote{The modified Bessel functions of the first kind is defined by the series representation $I_{\alpha}(x)
= \sum_{m \ge 0}\frac{1} {m!\Gamma(m+\alpha+1)}\left( \frac{x}{2}
\right)^{2m+\alpha}$, for a real number $\alpha$ which is not a
negative integer, and satisfies $I_{-n}(x) = I_n(x)$ for integer $n$}.
A similar computation for $y\leq 0$ shows that for all integer $y$, 
$$
f^1_t(y) = \frac{(T-t)I_{|y-1|}(2(T-t))}{I_{|y|}(2(T-t))},
$$
so that finally
$$
\lambda^{\bG,1}_t =
\frac{I_{|N_T-N_t-1|}(2(T-t))}{I_{|N_T-N_t|}(2(T-t))}\quad
\text{and}\quad \lambda^{\bG,2}_t =
\frac{I_{|N_T-N_t+1|}(2(T-t))}{I_{|N_T-N_t|}(2(T-t))}.
$$
%Note also that as $t\to T$, $f^1_t(y)\sim y$
%if $y>0$; $f^1_t(0)\sim (T-t)^2$ and $f^1_t(y)\sim
% \frac{1}{(|y|+1)}(T-t)^{2}$ if $y<0$. 

% while for $y\leq 0$ this expression remains
%bounded. Similarly,  $\frac{f^2_t(y)}{T-t}\sim \frac{|y|}{T-t}$ if
%$y<0$, and for $y\geq 0$ this expression remains bounded. 

%that since $I_\alpha(z)\sim
%\frac{1}{\Gamma(\alpha+1)}\left(\frac{z}{2}\right)^{\alpha}$ when
%$\alpha$ is not a negative integer, for $f^1_t(y) \sim y$ for $y>0$,
%$f^1_t(y) \sim (T-t)^2$ for $y = 0$ and $f^1(y)\sim y + $

\paragraph{Upper bound by duality}
From equation
(\ref{ex:poisson_density_F}), we get
\begin{align*}
\mathbb{E}^{\mathbb{P}}[1_{N_T = x}\log Z_T] & = \mathbb{E}^{\mathbb{P}}\left[ 1_{N_T = x}
\mathbb{E}^{\mathbb{P}} \left[ \left.  \sum_{i=1}^{2} \int\limits_0^T {\log \alpha^i_t dN^i_t
- \int\limits_0^T(\alpha^i_t - 1)dt} \right|  \mathcal{G}_0\right]  \right].
\end{align*}
Since we are only interested in an upper bound, we can restrict the
discussion to equivalent local martingale densities for which 
\begin{align}
\mathbb E\left[\int_0^t \log \alpha^i_s (dN^i_s - \lambda^{\bG,i}_s)
  \Big| \mathcal G_0\right]=0\quad \text{a.s.}\label{martcond}
\end{align}
In this case, the above expectation becomes
\begin{align}\label{ex:possion_expectation}
\mathbb{E}^{\mathbb{P}}[1_{N_T = x}\log Z_T] &= \mathbb{E}^{\mathbb{P}}\left[ 1_{N_T = x} \sum_{i=1}^{2}
\int\limits_0^T {\left( \lambda^{\mathbb{G},i}_t\log \alpha^i_t -
(\alpha^i_t - 1)\right) dt} \right]\\
& = \mathbb{E}^{\mathbb{P}}\left[ 1_{N_T = x} \int\limits_0^T {\left( \lambda^{\mathbb{G},1}_t\log \alpha^1_t -
\alpha^1_t + \lambda^{\mathbb{G},2}_t\log (e\alpha^1_t) -
e\alpha^1_t \right) dt} \right]\notag
\end{align}
The expression
$$
\lambda^{\mathbb{G},1}_t\log \alpha^1_t -
\alpha^1_t + \lambda^{\mathbb{G},2}_t\log (e\alpha^1_t) -
e\alpha^1_t 
$$
is concave as function of $\alpha^1_t$, 
with unique maximum attained at 
$$\frac{\lambda^{\mathbb{G},1}_t
+\lambda^{\mathbb{G},2}_t}{e+1}
$$
Fix $\varepsilon>0$ and let $\alpha^1_t = \frac{\lambda^{\mathbb{G},1}_t
+\lambda^{\mathbb{G},2}_t}{e+1} 1_{0\leq t\leq T-\varepsilon} +
1_{T-\varepsilon < t \leq T}$, and $\alpha^2_t = e\alpha^1_t$. These
values are bounded on $[0,T]$ implying that the corresponding density
satisfies $\mathbb{E}^{\mathbb{P}}[Z_T]=1 $ and \eqref{martcond}. 
Plugging it into (\ref{ex:possion_expectation}), we obtain
\begin{align*}
&\sup_{Z \in ELMMD(\mathbb{F},\mathbb{P})} \mathbb{E}^{\mathbb{P}} \left[ 1_{N_T = x} \log Z_T\right] \geq \mathbb{E}^{\mathbb{P}}\left[ 1_{N_T = x}  \int\limits_{T-\varepsilon}^T
{\left(-1 -e +  \lambda_t^{\bG,2} \right) dt}
\right]\\
&+  \mathbb{E}^{\mathbb{P}}\left[ 1_{N_T = x}  \int\limits_0^{T-\varepsilon}
{\left( \log \left( \frac{\lambda^{\mathbb{G},1}_t +
\lambda^{\mathbb{G},2}_t}{e+1} \right)  (\lambda^{\mathbb{G},1}_t +
\lambda^{\mathbb{G},2}_t) - \lambda^{\mathbb{G},1}_t + 2\right) dt}
\right].
\end{align*}
When $\varepsilon\to 0$, the expression under the two expectations is
bounded from below and 
converges monotonically to 
$$
1_{N_T = x}  \int\limits_0^{T}
{\left( \log \left( \frac{\lambda^{\mathbb{G},1}_t +
\lambda^{\mathbb{G},2}_t}{e+1} \right)  (\lambda^{\mathbb{G},1}_t +
\lambda^{\mathbb{G},2}_t) - \lambda^{\mathbb{G},1}_t + 2\right) dt},
$$ 
which shows that 
\begin{align}
&\sum_{x\in \mathbb Z}\sup_{Z \in ELMMD(\mathbb{F},\mathbb{P})} \mathbb{E}^{\mathbb{P}} \left[ 1_{N_T = x} \log Z_T\right] \notag\\
&\geq  \mathbb{E}^{\mathbb{P}}\left[\int\limits_0^{T}
{\left( \log \left( \frac{\lambda^{\mathbb{G},1}_t +
\lambda^{\mathbb{G},2}_t}{e+1} \right)  (\lambda^{\mathbb{G},1}_t +
\lambda^{\mathbb{G},2}_t) - \lambda^{\mathbb{G},1}_t + 2\right) dt}
\right],\label{dualbound}
\end{align}
and in particular the right-hand side is finite (by the duality
result). 
%Finally, the expectation in
%(\ref{eq:dual_second_term}) can be computed by numerical
%integration using the explicit formulas for the intensities $\lambda^{\mathbb{G},1},
%\lambda^{\mathbb{G},2}$ obtained in the previous paragraph.

\subsubsection{A lower bound by direct computation}
Let
$\pi^{\mathbb{G}}$ be a ${\mathbb{G}}$-predictable strategy, which
denotes the ratio invested in the risky asset, and let
$V^{1,\pi^{\mathbb{G}}}$ be the corresponding self-financing wealth
process that starts from $x=1$ and whose dynamics are
$$\frac{dV^{1,\pi^{\mathbb{G}}}_t}{V^{1,\pi^{\mathbb{G}}}_{t-}}
= \pi^{\mathbb{G}}_t \frac{dS_t}{S_{t-}} = \pi^{\mathbb{G}}_t \left( (e-1) dN^1_t + (e^{-1} - 1) dN^2_t\right).$$
The logarithm of $V^{1, \pi^{\mathbb{G}}}$ satisfies
\begin{align*}
d\log V_t^{1, \pi^{\mathbb{G}}} &=  \log \left( 1 + (e-1)\pi^{\mathbb{G}}_t \right) dN^1_t
+ \log \left( 1 + (e^{-1}-1)\pi^{\mathbb{G}}_t \right) dN^2_t\\
&= \left(   \log \left( 1 + (e-1)\pi^{\mathbb{G}}_t \right) \lambda^{\mathbb{G},1}_t
+ \log \left( 1 + (e^{-1}-1)\pi^{\mathbb{G}}_t \right)
\lambda^{\mathbb{G},2}_t\right) dt + \text{$\bG$-local martingale}.
\end{align*}
Since we are only looking for a lower bound, we can restrict the
discussion to strategies for which the coefficients in front of $N^1$
and $N^2$ are bounded and the local martingale above is a true
martingale. 
Taking the expectation of both sides, we obtain that
\begin{align}\mathbb{E}^{\mathbb{P}}\left[ \log V_T^{1,\pi^{\mathbb{G}}} \right]
= \mathbb{E}^{\mathbb{P}} \left[ \int\limits_0^T
{\left(   \log \left( 1 + (e-1)\pi^{\mathbb{G}}_t \right) \lambda^{\mathbb{G},1}_t
+ \log \left( 1 + (e^{-1}-1)\pi^{\mathbb{G}}_t \right) \lambda^{\mathbb{G},2}_t\right) dt} \right].\label{expport}\end{align}
The expression under the integral sign is concave in $\pi^\bG_t$ and
is maximized by 
$$ \frac{ \lambda^{\mathbb{G},1}_t(e-1)
+ \lambda^{\mathbb{G},2}_t(e^{-1}-1) } {
(e-1)(1-e^{-1})(\lambda^{\mathbb{G},1}_t + \lambda^{\mathbb{G},2}_t)
}\in \left(-\frac{1}{e-1} ,\frac{1}{1-e^{-1}}\right).$$
Fix $\varepsilon\in(0,1)$, define $\tau_\varepsilon = \inf\{t>0:
\lambda^{\mathbb{G},1}_t\notin(\varepsilon,\varepsilon^{-1})\
\text{or}\
\lambda^{\mathbb{G},2}_t\notin(\varepsilon,\varepsilon^{-1})\}\wedge
T$ and let
$$
\pi^\bG_t = \frac{ \lambda^{\mathbb{G},1}_t(e-1)
+ \lambda^{\mathbb{G},2}_t(e^{-1}-1) } {
(e-1)(1-e^{-1})(\lambda^{\mathbb{G},1}_t + \lambda^{\mathbb{G},2}_t)
}1_{t\leq \tau_\varepsilon}. 
$$
% The expected
%log-utility for the insider is now computed by using the formulas of
%$\lambda^{\mathbb{G},1}$ and $\lambda^{\mathbb{G},2}$ as in the
%previous subsection.

%$$
%1 + (e-1)\pi^{\mathbb{G}}_t = \frac{ (e+1)\lambda^{\mathbb{G},1}_t } {
%\lambda^{\mathbb{G},1}_t + \lambda^{\mathbb{G},2}_t
%},\qquad 1 + (e^{-1}-1)\pi^{\mathbb{G}}_t = \frac{ (1+e^{-1})\lambda^{\mathbb{G},2}_t } {
%\lambda^{\mathbb{G},1}_t + \lambda^{\mathbb{G},2}_t
%}
Substituting this into \eqref{expport}, we then get
\begin{align*}\mathbb{E}^{\mathbb{P}}\left[ \log V_T^{1,\pi^{\mathbb{G}}} \right]
&= \mathbb{E}^{\mathbb{P}} \left[ \int\limits_0^{\tau_\varepsilon}
{\left(    \lambda^{\mathbb{G},1}_t \log \frac{ (e+1)\lambda^{\mathbb{G},1}_t } {
\lambda^{\mathbb{G},1}_t + \lambda^{\mathbb{G},2}_t
}
+ \lambda^{\mathbb{G},2}_t\log \frac{ (1+e^{-1})\lambda^{\mathbb{G},2}_t } {
\lambda^{\mathbb{G},1}_t + \lambda^{\mathbb{G},2}_t
} \right) dt} \right]\\
& = - \mathbb{E}^{\mathbb{P}} \left[ \int\limits_0^{\tau_\varepsilon}
{\left(    (\lambda^{\mathbb{G},1}_t +\lambda^{\mathbb{G},1}_t)  \log \frac {
\lambda^{\mathbb{G},1}_t + \lambda^{\mathbb{G},2}_t
}{ e+1 } - \lambda^{\mathbb{G},1}_t + 2\right)dt}\right]\\ &+ \mathbb{E}^{\mathbb{P}} \left[ \int\limits_0^{\tau_\varepsilon}
{\left( \lambda^{\mathbb{G},1}_t \log \lambda^{\mathbb{G},1}_t +
    \lambda^{\mathbb{G},2}_t \log (\lambda^{\mathbb{G},1}_t)- \lambda^{\mathbb{G},1}_t -\lambda^{\mathbb{G},2}_t+2\right) dt} \right].
\end{align*}
When $\varepsilon\to 0$, $\tau_\varepsilon\to T$ and the first
expectation above clearly converges to \eqref{dualbound}. To finish
the proof, it remains to check that 
% $$
% \mathbb{E}^{\mathbb{P}} \left[ \int\limits_0^{T}
% {\left( \lambda^{\mathbb{G},1}_t \log \lambda^{\mathbb{G},1}_t +
%     \lambda^{\mathbb{G},2}_t \log (\lambda^{\mathbb{G},1}_t)-
%     \lambda^{\mathbb{G},1}_t -\lambda^{\mathbb{G},2}_t+2\right) dt}
% \right] = - \sum_{ x \in \mathbb{Z}} \mathbb{P}[N_T = x] \log \mathbb{P}[N_T = x],
% $$
% or in other words that 
$$
\mathbb{E}^{\mathbb{P}} \left[ \int\limits_0^{T}
{\left( \lambda^{\mathbb{G},1}_t \log \lambda^{\mathbb{G},1}_t +
    \lambda^{\mathbb{G},2}_t \log (\lambda^{\mathbb{G},1}_t)-
    \lambda^{\mathbb{G},1}_t -\lambda^{\mathbb{G},2}_t+2\right) dt}
\Big| N_T = x\right] = - \log \mathbb{P}[N_T = x],
$$
for $x\in \mathbb Z$. 

To this end, let $M_t = \mathbb P[N_T = x|\mathcal
F_t]$. By direct computation, we obtain that $M_t = e^{-2(T-t) }
I_{|x-N_t|}(2(T-t))$. The change of variable formula then yields
$$
M_T = M_0 +\int_0^T e^{-2(T-t)}(I_{|x-N_t-1|} - I_{|x-N_t|})(dN^1_t
- dt) + \int_0^T e^{-2(T-t)}(I_{|x-N_t+1|} - I_{|x-N_t|})(dN^2_t
- dt),
$$
and further, on the event $N_T = x$, 
$$
\log M_T = \log M_0 +\int_0^T \log \lambda^{\bG,1}_tdN^1_t
 + \int_0^T \log \lambda^{\bG,2}_t dN^2_t - \int_0^T\left(\lambda^{\bG,1}_t+ \lambda^{\bG,2}_t-2\right)dt. 
$$
Finally, we may conclude that 
\begin{align*}
&-\log \mathbb P[N_T = x] = -\log M_0 = -\mathbb E[\log M_0 | N_T =
x]\\
& =  \mathbb E\left[\int_0^T \log \lambda^{\bG,1}_tdN^1_t
 + \int_0^T \log \lambda^{\bG,2}_t dN^2_t -
 \int_0^T\left(\lambda^{\bG,1}_t+ \lambda^{\bG,2}_t-2\right)dt
 \Big|N_T = x\right]\\
& = \mathbb E\left[\int_0^T \log \lambda^{\bG,1}_t \lambda^{\bG,1}_t dt
 + \int_0^T \log \lambda^{\bG,2}_t \lambda^{\bG,2}_t dt -
 \int_0^T\left(\lambda^{\bG,1}_t+ \lambda^{\bG,2}_t-2\right)dt
 \Big|N_T = x\right]
\end{align*}
provided we may show that 
$$
\mathbb E\left[\int_0^T \log^2 \lambda_t^{\bG,1} \lambda_t^{\bG,1}
  dt\right]<\infty. 
$$
To see this, observe that from the asymptotics of the Bessel function
it follows that 
$$
|\log \lambda^{\bG,1}_t|\leq C |\log (T-t)| \log (1+|N_T-N_t|).
$$
Plugging this into the above estimate and using \eqref{intenslambda}, we finally have 
$$
\mathbb E\left[\int_0^T \log^2 \lambda_t^{\bG,1} \lambda_t^{\bG,1} dt\right]\leq
  \mathbb E\left[\int_0^T C^2 |\log^2 (T-t)| \log^2 (1+|N_T-N_t|) \frac{N^1_T-N^1_t}{T-t}
  dt\right],
$$
which is easily shown to be bounded (use Cauchy-Schwarz inequality
plus the estimate $\mathbb E[N^\alpha_T] = O(T)$ for $\alpha>0$). }

\section{Initial enlargement with a general random variable}\label{section:G_continuous}

We now consider the case when the $\F_T-$measurable random variable
$G$, which represents the insider information, is not purely atomic
as it was the case in the previous section.

It is usually observed that in this case the value of logarithmic
utility of the insider is infinite, as for example in Theorem 4.4 of
\cite{pikovsky_anticipative_1996} where the insider has exact information about at least one
stock's terminal price, or in \cite{amendinger_additional_1998} where  the insider's additional
expected logarithmic utility is related to the entropy of $G$. This
difficulty appears at $T$, the time when the conditional law of $G$
given $\mathcal{F}_T$ is a Dirac measure and hence Jacod's condition
fails.

Since, as we saw before, NUPBR is the minimal condition for
well-posed expected utility maximization problems, it is useful to
have conditions for NUPBR to hold. {\cite{acciaio_arbitrage_2014} and \cite{aksamit2015optional} give a sufficient
condition so that NUPBR holds under $\mathbb{G}$ in infinite time
horizon settings.} Their idea is that, if the processes $p^x$ and $S$
do not jump to zero at the same time, then one can construct an
equivalent (local) martingale deflator (ELMD) under $\mathbb{G}$ and
it is known that the existence of such an ELMD implies NUPBR.
However, in finite horizon settings, it may happen that the process
$p^x$ is not well-defined at $T$, making it impossible to define an
ELMD because $p^x$ appears in the denominator of such an ELMD.

In this paper we consider a finite horizon $T>0$ and, in the present
section, we shall study arbitrage properties and expected utility
maximization for the case of a non-atomic $G$. Before coming to
utility maximization, in the next subsection we show that
if $G$ is non-atomic and the set of local martingale densities
$ELMMD(\mathbb{F},\mathbb{P})$ is uniformly integrable, then there
always exists an arbitrage of the first kind and so NUPBR fails. It shows in particular that, if the
$(\mathbb{F}, \mathbb{P})$ market is complete, then NUPBR always
fails under $\mathbb{G}$.
This negative message implies that the non-uniform integrability of
$ELMMD(\mathbb{F}, \mathbb{P})$ is a necessary condition for NUPBR
under $\mathbb{G}$. This result is then accompanied by two examples:
one where the set of local martingale densities
$ELMMD(\mathbb{F},\mathbb{P})$ is uniformly integrable, and one where
it is not. In the second example 
the expected log-utility is finite, which gives (see Proposition
\ref{L-NUPBR}) a sufficient condition for
NUPBR to hold. On the other hand, in this second example
the insider has non-scalable arbitrage opportunities so that
NFLVR cannot hold. Initial filtration
enlargement with a non-atomic $\mathcal F_T$-measurable random
variable may therefore lead to viable market models which allow for economically
meaningful unscalable arbitrages.

\subsection{Arbitrage of the first kind}\label{section:A1_G}
The following result does not require Assumption
\ref{assum:density_jump}. 
\begin{prop}\label{pro:A1_G}(Arbitrage of the first kind)
Assume that
\begin{itemize}
\item The law of $G$ is not purely atomic,
\item The set of densities of equivalent local martingale measures $ELMMD(\mathbb{F},\mathbb{P})$ is uniformly integrable.
\end{itemize}
Then there exists an arbitrage of the first kind for the insider.
\end{prop}
\begin{proof}
Let us choose $B \subset \mathbb{R}$ such that $B$ does not contain
any atoms of $G$ and $\mathbb{P}[G\in B] = c > 0$. For each $n$, let
$(B^n_i)_{1 \le i \le n}$ be a partition of $B$ such that
$\mathbb{P}[G \in B^n_i] = c/n$. We are going to show that $1_{G
\in B}$ is an arbitrage of the first kind, in the
sense of Definition \ref{D3}. First,
consider the superhedging price of $1_{G \in B^n_i}$ and its
associated hedging strategy $H^{\mathbb{F},i}$ in the $(\mathbb{F},
\mathbb{P})$-market (see Corollary 10 in \cite{delbaen_general_1994}), for which
$$\sup_{Z \in ELMMD(\mathbb{F}, \mathbb{P})}\mathbb{E}^{\mathbb{P}}[Z_T1_{G \in B^n_i}] + (H^{\mathbb{F},i} \cdot S)_T \ge 1_{G \in B^n_i}.$$
Therefore,
\begin{equation}\label{eq:A1_price}
\sum_{i=1}^n \sup_{Z \in ELMMD(\mathbb{F}, \mathbb{P})}
\mathbb{E}^{\mathbb{P}}[Z_T1_{G \in B^n_i}]1_{G \in B^n_i} + \left(\left(
\sum_{i = 1}^n H^{\mathbb{F},i}1_{G\in B^n_i} \right)  \cdot S\right)_T \ge
1_{G \in B}.
\end{equation}
Because the set of all local martingale densities $\{ Z_T: Z \in
ELMMD(\mathbb{F},\mathbb{P}) \}$ is uniformly integrable, for any
$\varepsilon > 0$ there exists $K > 0$ such that
$$\sup_{Z \in ELMMD(\mathbb{F},\mathbb{P})} \mathbb{E}^{\mathbb{P}}[Z_T 1_{Z_T>K} ] \le \varepsilon.$$
The initial capital in (\ref{eq:A1_price}) can then be estimated by
\begin{align*}
& \sum_{i=1}^n \sup_{Z \in ELMMD(\mathbb{F}, \mathbb{P})} \left(
\mathbb{E}^{\mathbb{P}}
[Z_T1_{Z > K}1_{G \in B^n_i}] + \mathbb{E}^{\mathbb{P}}[Z_T1_{Z_T \le K}1_{G\in B^n_i}] \right) 1_{G \in B^n_i}\\
&\le \sum_{i=1}^n \left( \varepsilon +  K \mathbb{P}[G \in B^n_i]
\right)1_{G \in B^n_i}= \sum_{i=1}^n \left( \varepsilon +  K
\frac{c}{n} \right)1_{G \in B^n_i}.
\end{align*}
We can choose $\varepsilon$ and $n$ such that the initial capital in
(\ref{eq:A1_price}) is arbitrarily small and thus the random
variable $1_{G \in B}$ is an arbitrage of the first kind.
\end{proof}
\begin{rem}[A comparison with \cite{amendinger_additional_1998}]
{The above proposition extends Theorem 4.4 of
\cite{amendinger_additional_1998}. More precisely,
\cite{amendinger_additional_1998} show that the insider's additional
expected logarithmic utility up to time $T$ becomes infinite (which
implies that NUPBR fails). However, their results apply only to
continuous processes and require an even stronger condition than
market completeness, namely that the inverse of $p^G$ may be
represented as a stochastic integral, see condition (45) therein. By
our result, we are able to construct unbounded profits in general
market settings. In particular, the following example shows that the
property of uniform integrability may also hold for some incomplete
market models.}
\end{rem}

\begin{ex}[Incomplete market with uniformly integrable set of
    equivalent martingale densities]
We consider a risky asset whose (discounted) price evolves as
$$dS_t = S_{t-} \sigma(t) (\theta dN^1_t + (1-\theta) dN^2_t - dt)$$
where $\theta \in (0,1)$ and $\sigma(t)$ is a 
continuous function which is not constant on $[0,T]$. The filtration $\mathbb{F}$ is generated by the
two independent standard Poisson processes $N^1$ and $N^2$. Any
martingale density has the form \eqref{ex:poisson_density_F}
where now $\alpha^1$ and $\alpha^2$ are positive predictable
integrable processes
satisfying $ $ $\theta \alpha^1_t + (1-\theta)\alpha^2_t = 1,
\mathbb{P}-a.s.$ Therefore,
$$0 \le \alpha^1 \le \frac{1}{\theta}, \qquad 0 \le \alpha^2 \le \frac{1}{1-\theta}. $$
These inequalities lead to an upper bound for all martingale
densities:
$$Z_T = e^{- \int_0^T {(\alpha^1_t + \alpha^2_t -2)dt} } \prod_{i=1}^{N^1_T} \alpha^1_{t_i}
\prod_{j=1}^{N^2_T} \alpha^2_{t_j} \le e^{2T}
\frac{1}{\theta^{N^1_T}} \frac{1}{(1-\theta)^{N^2_T}}.$$
 As a result, the set of martingale
densities is uniformly integrable. Let $(T^1_i)_{i \ge 1}$ and $
(T^2_j)_{j \ge 1}$ be the jump times of $N^1$ and $N^2$
respectively.  Because $\sigma$ is a continuous function, the random
variables  $\sigma(T^1_i), \sigma(T^2_j)$ are continuous  and not constant. This means
that the random variable
$$G = S_T = \exp\left( - \int\limits_0^T {\sigma(s)ds}  \right)
\prod_{i = 1}^{N^1_T} \left( 1+ \theta \sigma(T^1_i) \right)
\prod_{j = 1}^{N^2_T} \left( 1+ (1-\theta) \sigma(T^2_j) \right) $$
is a non-atomic random variable for which the market of the insider
does not satisfy NUPBR.
\end{ex}

In the literature there are examples of incomplete market models where
NUPBR holds and therefore the set of equivalent martingale densities
is not uniformly integrable. The following one is due to \cite{kohatsu-higa_insider_2011}.
\begin{ex}[Incomplete market where NUPBR holds but NFLVR
  fails]
The market model in this example consists of a risk-free asset paying
zero interest and a single risky asset driven by a L\'evy process with two-sided jumps. The public information $\mathbb{F}$ is the natural
filtration generated by a Brownian motion $W$ and two independent
Poisson processes $N^1, N^2$ with common intensity $\lambda = 1$. The risky
asset is $S_t = \exp(M_t)$ where $M_t = W_t + N^1_t - N^2_t $ is a
$\mathbb{F}$-martingale. The dynamics of $S$ under $\mathbb{F}$ is
\begin{align*}
dS_t = S_{t-} \left( dW_t + \frac{1}{2} dt + (e - 1)dN^1_t + (e^{-1}
- 1)dN^2_t\right). \end{align*} Let $H^{\mathbb{F}}$ be an
$\mathbb{F}$- self-financing strategy and denote by
$\pi^{\mathbb{F}}$ the fraction of wealth invested in the stock. The
associated wealth is $V^{v,\pi^{\mathbb{F}}}$ satisfying
$$ \frac{dV^{v,\pi^{\mathbb{F}}}_t}{V^{v,\pi^{\mathbb{F}}}_{t-}} = \pi^{\mathbb{F}}_t \frac{dS_t}{S_{t-}}, \qquad V^{v,\pi^{\mathbb{F}}}_0 = v.$$
The strategy $\pi^{\mathbb{F}}$ is admissible if for all $t \in
[0,T]$ we have that $V_t^{v,\pi^{\mathbb{F}}} \ge 0,
\mathbb{P}-a.s.$
The insider has
the additional information given by the final value of $S$ so that
$\mathcal{G}_t = \mathcal{F}_t \vee \sigma(S_T) = \mathcal{F}_t \vee
\sigma(M_T)$. 
In \cite{kohatsu-higa_insider_2011} it is shown that all admissible strategies are bounded and the expected logarithmic
utility for an insider is bounded from above.
This entails that the $(\mathbb{G},
\mathbb{P})$-market satisfies NUPBR (see Proposition \ref{L-NUPBR}). Furthermore, the insider has arbitrage
opportunities since he knows the final value of $S$. For example, if
he knows that $S_T > 1$, which happens with positive probability, he
could buy the asset $S$ and hold it until maturity; being $S_0=1$,
this implies a riskless profit.
\end{ex}
\subsection{An approximation procedure} \label{section:approx_scheme}

Assume that we are  given a non purely atomic random variable $G\in\R$
representing the information of the insider. Let $\{ \Gamma^n_i, i =
1,...,n\}$ be a finite increasing partition of $\mathbb{R}^+$ and
denote $$\sigma(G^n) = \sigma \left( \{G \in \Gamma^n_i \}, i =
1,...,n\right).$$ We approximate $\sigma(G)$ by the increasing
sequence of sigma algebras $\sigma(G^n)$
$$\sigma(G^n) \subset \sigma (G^{n+1}) \subset ... \sigma(G) = 
\sigma\left( \bigcup_{n \ge 1} \sigma(G^n) \right)$$
and define an increasing sequence of filtrations $\mathbb{G}^n =
(\mathcal{G}^n_t)_{t \in [0,T]}$, where $\mathcal{G}^n_t =
\bigcap_{\varepsilon > 0} \mathcal{F}_{t+ \varepsilon} \vee
\sigma(G^n)$. For each $n$, we shall
use the results from Section \ref{S2} to compute the expected utility under $\mathbb{G}^n$, the information at
level $n$. The convergence properties given below in this section
will then enable us to obtain the corresponding
results under $\mathbb{G}$.

First we have to make sure that for each $\mathbb{G}^n-$market we can indeed 
apply the results from Section \ref{S2}
which are shown under Assumption \ref{assum:density_jump}. Since in the $\mathbb{G}^n-$market the
insider's information is given by
$\sigma\left(\{G\in\Gamma_i^n\},\,i=1,\cdots,n\right)$, we need to
show that for all integers $n$ and $i$ the quantities
$\frac{\mathbb{P}[G\in\Gamma_i^n\mid\F_t]}{\mathbb{P}[G\in\Gamma_i^n]}$
do not jump to zero.
This is the object of the following assumption.

\begin{ass}\label{ass:ab}
For every $a < b$, the
$\mathbb{F}$-martingale
$$p^{(a,b)}_t(\omega) :=\frac{\mathbb{P}[G \in (a,b)|\mathcal{F}_t](\omega)}{\mathbb{P}[G \in (a,b)]}
= \frac{1}{\mathbb{P}[G \in (a,b)]} \int\limits_a^b
{p^x_t(\omega)\mathbb{P}[G \in dx]}$$ does not jump to zero.
\end{ass}

{The main result of this subsection is now the following convergence result
\begin{prop}\label{P-Conv} For any $\mathbb{G}$-predictable strategy
$H^{\mathbb{G}}$, there exists a sequence of
$(\mathbb{G}^n)$-predictable strategies $(H^{\mathbb{G}^n})_n$ such
that for every $t\in [0,T]$, $H^{\mathbb{G}^n}_t \to H^{\mathbb{G}}_t$
almost surely.
\end{prop} 
\begin{proof}
By definition of the predictable $\sigma$-field, we can find a
sequence of c\`agl\`ad $\mathbb{G}$-adapted processes
$(H^{\mathbb{G},n})_n$ such
that for every $t\in [0,T]$, $H^{\mathbb{G},n}_t \to H^{\mathbb{G}}_t$
almost surely. Moreover, without loss of generality, the processes
$(H^{\mathbb{G},n})_n$ may be assumed to be bounded. 

Let $\{T^n_i\}_{0\leq i \leq n}^{n\geq 1}$ be a
sequence of (deterministic) partitions of the interval $[0,T]$ such that for every
$n$, 
$$
0 = T^n_0 < T^n_1 < \dots < T^n_n = T
$$
and $\lim_{n\to \infty}\max_{0\leq i \leq n-1}|T^n_{i+1}-T^n_i| =
0$. For a c\`agl\`ad $\mathbb{G}$-adapted
process $X$, define the process $X^n$ by $X^n_0 = X_0$ and 
$$
X^n_t =  \sum_{i=0}^{n-1} X_{T_i}1_{]T_i,T_{i+1}]}(t),
$$
for $0<t\leq T$. Then,
$X^n$ is $\mathbb G$-adapted and 
for every $t\in [0,T]$, $X^n_t \to X_t$ almost surely. 

Finally, let $Y$ be a process of the form
$$
Y_t = \sum_{i=0}^{k} y_{T_i}1_{]T_i,T_{i+1}]}(t),
$$
where, for each $i$, $y_{T_i}$ is a bounded $\mathcal G_{T_i}$
  measurable random variable, and $T_i$ is a deterministic time such
  that $0 = T_0< T_1< \dots < T_n=T$. 
We define
$$Y^{\mathbb{G}^n}_t: = \sum_{i=1}^{k} \mathbb{E}[y_{T_{i}}|\mathcal{G}^n_{T_i}] 1_{]T_i, T_{i+1}]} (t).$$
The process $Y^{\mathbb{G}^{n}}$ is $\mathbb{G}^n$-predictable.
Using L\'evy's ``Upward'' Theorem (see Theorem II.50.3 of \cite{rogers1979diffusion}), we
obtain $\mathbb{E}[y_{T_{i}}|\mathcal{G}^n_{T_i}] \to y_{T_{i}},
\mathbb{P}-a.s.$, which means that for every $t\in [0,T]$, $
Y^{\mathbb{G}^n}_t \to Y_t$ almost surely. Combining the three
approximations described above gives the proof of the proposition.

\end{proof}}

\subsection{Expected utility maximization}\label{section:NUPBR_continuous_G}

Analogously to subsection \ref{S2.2} we consider now the expected
utility maximization and a dual representation of the optimal
value for the case of a non purely atomic $G$. Based on Proposition
\ref{L-NUPBR}, this will then allow one to
show NUPBR under Assumption \ref{assum:density_jump} by using the
finiteness of expected utility.

We start with the main convergence result
\begin{thm}\label{thm:log_converge}
Let $U:\mathbb R\to \mathbb R_+$ be increasing and continuous
and suppose that $S$ is a continuous semimartingale.
Then,
$$\lim_{n \to \infty} \sup_{H \in \mathcal{A}^{\mathbb{G}^n}_1}\mathbb{E}^{\mathbb{P}}
[U( V^{1, H}_T)] = \sup_{H \in
\mathcal{A}^{\mathbb{G}}_1}\mathbb{E}^{\mathbb{P}} [U( V^{1,
H}_T)].$$
\end{thm}
\begin{proof}
The inequality $\leq$ is trivial since $U$ is increasing and $\mathcal{A}^{\mathbb{G}^n}_1 \subset\mathcal{A}^{\mathbb{G}}_1$.
To prove the opposite inequality, we choose $H^{\mathbb{G}}\in
\mathcal{A}^{\mathbb{G}}_1$.
Our first aim is to show that
$H^{\mathbb G}$ may be supposed to be bounded. Indeed, the stochastic
integral $H^\mathbb G \cdot S$ is defined as the $\mathcal H^2$-limit
of $H \mathbf 1_{|H|\leq n} \cdot S = H^+ \mathbf 1_{|H|\leq n} \cdot
S -  H^- \mathbf 1_{|H|\leq n} \cdot
S$, see \cite[Section IV.2]{protter_stochastic_2003}, and by the dominated convergence theorem for stochastic integrals
\cite[Theorem IV.32]{protter_stochastic_2003},
this limit is attained uniformly in compacts in probability, and hence
also almost surely along a subsequence. Introduce the stopping time
$$
\tau^n_\varepsilon = \inf\{t>0: 1+\varepsilon + \int_0^t H \mathbf
1_{|H|\leq n}  dS\leq 0 \}\wedge T. 
$$
Observe that 
$$
1+\varepsilon + \int_0^{t\wedge \tau^n_\varepsilon} H \mathbf
1_{|H|\leq n}  dS\geq 0,\quad t\in [0,T],
$$
which means that the strategy defined by $\frac{1}{1+\varepsilon} H\mathbf
1_{|H|\leq n} \mathbf 1_{t\leq \tau^n_\varepsilon}$ belongs to
$\mathcal A^{\mathbb G}_1$. 
The uniform convergence of the stochastic integrals implies that
$\mathbb P[\tau^n_\varepsilon = T] \to 1$ as $n\to \infty$, and also
$$
1+ \int_0^T \frac{1}{1+\varepsilon} H\mathbf
1_{|H|\leq n} \mathbf 1_{t\leq \tau^n_\varepsilon} dS \to
1+\frac{1}{1+\varepsilon}\int_0^T H_t dS_t
$$
as $n\to \infty$. By Fatou's lemma we then have
$$
\mathbb E\left[U\left(1+\frac{1}{1+\varepsilon}\int_0^T H_t
    dS_t\right)\right]\leq \liminf_n \mathbb E\left[U\left(1+ \int_0^T \frac{1}{1+\varepsilon} H\mathbf
1_{|H|\leq n} \mathbf 1_{t\leq \tau^n_\varepsilon} dS\right)\right],
$$
and another application of Fatou's lemma shows that 
$$
\mathbb E\left[U\left(1+\int_0^T H_t
    dS_t\right)\right]\leq \liminf_{\varepsilon\downarrow 0}\mathbb E\left[U\left(1+\frac{1}{1+\varepsilon}\int_0^T H_t
    dS_t\right)\right].
$$
This argument shows that
$$
\sup_{H \in
\mathcal{A}^{\mathbb{G}}_1}\mathbb{E}^{\mathbb{P}} [U( V^{1,
H}_T)] = \sup_{H \in
\overline{\mathcal{A}}^{\mathbb{G}}_1}\mathbb{E}^{\mathbb{P}} [U( V^{1,
H}_T)],
$$
where $\overline{\mathcal{A}}^{\mathbb{G}}_1$ denotes the strategies
in $\mathcal{A}^{\mathbb{G}}_1$ which are bounded by a deterministic
constant. Therefore, from now on we may (and will) assume $H$ to
be bounded by a constant $C$.

Now let $(H_n)$ be a sequence of strategies approximating $H$ in the
sense of Proposition \ref{P-Conv}, which may be assumed to be bounded
by the same constant $C$. Once again, by the dominated convergence
theorem for stochastic integrals, we show that $H_n\cdot S$ converges
to $H\cdot S$ uniformly on compacts in probability, and hence also
almost surely along a subsequence. Similarly to the previous part, we
construct an admissible strategy from $H_n$ by stopping it at a
suitable stopping time. An application of Fatou's lemma then shows
that 
$$
\liminf_{n \to \infty} \sup_{H \in \mathcal{A}^{\mathbb{G}^n}_1}\mathbb{E}^{\mathbb{P}}
[U( V^{1, H}_T)] \geq \sup_{H \in
\mathcal{A}^{\mathbb{G}}_1}\mathbb{E}^{\mathbb{P}} [U( V^{1,
H}_T)].
$$
Observing that the expression under the limit in the left hand side is
increasing in $n$ and combining this with the opposite inequality, we
conclude the proof. 
\end{proof}

\begin{rem}
Extending the result of Theorem \ref{thm:log_converge} to discontinuous processes seems to be a
difficult task, see in particular Example 11.2.6 in
\cite{ankirchner_information_2005}. One may, for example, obtain such
an extension under the following condition which is not easy to verify
in practice:
\begin{itemize}
\item[]For every
  $\varepsilon>0$ there exists a strategy $H\in \mathcal A^{\mathbb
    G}_1$ with 
$$
\mathbb{E}^{\mathbb{P}} [U( V^{1,
H}_T)] \geq \sup_{H \in
\mathcal{A}^{\mathbb{G}}_1}\mathbb{E}^{\mathbb{P}} [U( V^{1,
H}_T)] - \varepsilon,
$$ 
and, for a number $n\geq 1$, a sequence of $\mathbb
  G^n$-predictable integrable processes $(H^m)$ as well as a sequence of $\mathbb
  G^n$-stopping times $(\tau_m)$ with $\mathbb P[\tau_m = T]\to 1$,
  such that for every $m\geq 1$,
$$
\left|\int_0^t (H^m_t - H_t) dS_t\right| \leq \varepsilon
$$ 
for all $t\leq \tau_m$. 
\end{itemize}
\end{rem}

As a corollary to Theorems \ref{thm:log_converge} and
\ref{thm:gen-dual}, we obtain a dual
representation for the utility maximization problem in the case of
non-atomic $G$. 
\begin{thm}
Let $S$ be a continuous semimartingale, suppose that Assumptions
\ref{basic-assum} and \ref{ass:ab} hold true, and assume that 
 \begin{itemize}
\item[(i)] The function $U:(0,\infty)\to \mathbb R$ is strictly
  concave, increasing, continuously differentiable and satisfies the Inada conditions at
  $0$ and $\infty$.  
\item[(ii)] For every $y\in (0,\infty)$, there exists $Z
  \in 
  ELMMD(\bF,\bP)$ with $\mathbb E^{\mathbb P}[V(yZ)]<\infty$, where
  $V(y) =  \sup_x (U(x)-xy) $.
 \end{itemize}
Then, 
$$
\sup_{H\in \mathcal A^{\bG}_1} \mathbb E^{\bP} [U(V^{1,H}_T)] =
\lim_{n\to \infty}\sum_{i}
\inf_{y>0}\left\{y+
\inf_{Z \in ELMMD(\mathbb{F},\mathbb{P})} \mathbb{E}^{\mathbb{P}}
\left[ V\left(yZ_T \right)1_{G \in \Gamma^n_i}\right]\right\}.
$$
\end{thm}
In the case of logarithmic utility, a more explicit expression may be
obtained. 
\begin{cor}\label{pro:log_dual_cont}
Let $S$ be a continuous semimartingale, suppose that Assumptions
\ref{basic-assum} and \ref{ass:ab} hold true, and assume that $G$ has a continuous density $f(x)$ and a finite
entropy and that there exists $Z\in ELMMD(\mathbb F, \mathbb P)$ with
$\mathbb E^{\mathbb P}[\log Z_T]>-\infty$. 
Then the insider's expected log-utility is
\begin{align}\label{eq:log_dual_cont}
\sup_{H \in \mathcal{A}^{\mathbb{G}}_1} \mathbb{E}^{\mathbb{P}}[\log V^{1, H}_T] &= - \int {f(x)\log f(x) dx} \nonumber  \\
&+ \lim_{n \to \infty}  \sum_{i = 1}^n \left( - \log|\Gamma^n_i|
\mathbb{P}[G \in \Gamma^n_i] + \inf_{Z \in
ELMMD(\mathbb{F},\mathbb{P})} \mathbb{E}^{\mathbb{P}} \left[ 1_{G \in
\Gamma^n_i} \log \frac{1}{Z_T}\right] \right)  .
\end{align}
\end{cor}
\begin{proof}
Theorem \ref{thm:log_converge} and Corollary \ref{cor:log:dual} show
that \beq\label{logdual} \sup_{H \in \mathcal{A}^{\mathbb{G}}_1}
\mathbb{E}^{\mathbb{P}}[\log V^{1, H}_T] = \lim_{n \to \infty}
\sum_{i= 1}^n \left(  - \mathbb{P}[G \in \Gamma^n_i] \log
\mathbb{P}[G \in \Gamma^n_i] + \inf_{Z \in
ELMMD(\mathbb{F},\mathbb{P})} \mathbb{E}^{\mathbb{P}} \left[ 1_{G \in
\Gamma^n_i} \log \frac{1}{Z_T}\right] \right) \eeq
Now we consider
the first term in the right hand side. Using the mean value theorem,
we have that $\mathbb{P}[G \in \Gamma^n_i] = f(x^n_i)|\Gamma^n_i|$
for some $x^n_i \in \Gamma^n_i$. Thus,
\begin{align*}
- \mathbb{P}[G \in \Gamma^n_i] \log \mathbb{P}[G \in \Gamma^n_i] &= - \mathbb{P}[G \in \Gamma^n_i] \log (f(x^n_i)|\Gamma^n_i|)\\
 & = - f(x^n_i)\log f(x^n_i)  |\Gamma^n_i| - \mathbb{P}[G \in \Gamma^n_i] \log|\Gamma^n_i|.
\end{align*}
Letting $n$ tend to infinity, we get the result.
\end{proof}
As a consequence, the insider's log-utility problem is finite if $G$
has finite entropy and for every event $\{ G \in \Gamma^n_i \}$,
there exists a martingale density $Z_T$ such that the quantity
$\mathbb{E}^{\mathbb{P}}[1_{G \in \Gamma^n_i}\log(1/Z_T)] $ can
compensate the term $-\log|\Gamma^n_i|\mathbb{P}[G \in \Gamma^n_i]$.
In complete markets, it is impossible to find such a martingale
density for each event, implying that expected log-utility of the
insider is infinite. In incomplete markets, the result provides us
with a new criterion for NUPBR under $\mathbb{G}$ as stated in the
following
\begin{cor}\label{cor:crit1}
Under Assumption \ref{assum:density_jump}, if there exists a
constant $C<\infty$ such that for all $a$ and all $\varepsilon>0$
small enough,
\begin{align}\label{crit1}
\sup_{Z \in ELMMD}\mathbb E[1_{G \in (a,a+\varepsilon) }\log Z_T ]
\geq &-\mathbb P[G \in (a,a+\varepsilon)] \nonumber
\log \mathbb P[G \in (a,a+\varepsilon)] \\
&- C\mathbb P[G \in (a,a+\varepsilon)]
\end{align}
then the condition NUPBR holds under $\mathbb{G}$.
\end{cor}
\begin{proof} Consider (\ref{logdual}) for partitions of the form
$\Gamma^n_i=(a_i,a_i+\varepsilon_n)$ with $\varepsilon_n\downarrow
0$. Using then $(\ref{crit1})$ in (\ref{logdual}) one obtains
\begin{align*}
\sup_{H \in \mathcal{A}^{\mathbb{G}}_1} \mathbb{E}^{\mathbb{P}}[\log
V^{1, H}_T] &= \lim_{n \to \infty} \sum_{i= 1}^n \left(  -
\mathbb{P}[G \in \Gamma^n_i] \log \mathbb{P}[G \in \Gamma^n_i] + \inf_{Z \in
ELMMD(\mathbb{F},\mathbb{P})}
\mathbb{E}^{\mathbb{P}} \left[ 1_{\{G \in \Gamma^n_i\}} \log \frac{1}{Z_T}\right]  \right)\\
& \le \lim_{n \to \infty} \sum_{i= 1}^n C \mathbb{P}[G \in \Gamma^n_i] =
C.
\end{align*}
The expected log-utility of the insider is bounded and hence, by
Proposition \ref{L-NUPBR}, the condition NUPBR holds under
$\mathbb{G}$.
\end{proof}

\section*{Compliance with Ethical Standards}
The authors declare that they do not have any conflicts of interest
in relation to the present work. 

\section*{Appendix}

\begin{lem}\label{lem:ratio_dis_exp}
Assume that $X, Y$ are two independent exponential random variables
with parameters $\alpha, \beta$, respectively. Then the random
variable $Z = \frac{\alpha X}{\beta Y}$ has density $1/(1+ z)^2.$
\end{lem}
\begin{proof}
For $z > 0$. we compute the cumulative distribution of $Z$
\begin{align*}
\mathbb{P}[Z \le z] &= \mathbb{P}\left[   Y \ge \frac{\alpha X}{\beta z}\right]
= \int\limits_0^{\infty} {\left( \int\limits_{(\alpha x) / (\beta z)}^{\infty} {\beta e^{-\beta y}dy} \right) \alpha e^{-\alpha x}dx}\\
&= \int\limits_0^{\infty} { e^{\frac{-\alpha x}{z}} \alpha
e^{-\alpha x}dx}= \frac{z}{1 + z}.
\end{align*}
The density of $Z$ is obtained by taking derivative of the
cumulative distribution of $Z$ with respect to $z$.
\end{proof}
\begin{defn}[Optional projection - Definition 5.2.1 of \cite{jeanblanc_mathematical_2009}]
Let $X$ be a bounded (or positive) process, and $\mathbb{F}$ a given
filtration. The optional projection of $X$ is the unique optional
process ${}^{o}X$ which satisfies 
$$
\mathbb{E}[X_{\tau}1_{\tau <\infty}] ={}^{o}X_{\tau}1_{\tau < \infty}
$$
almost surely for any
$\mathbb{F}$-stopping time $\tau.$
\end{defn}
The following result helps us to find the compensator of a process
when passing to smaller filtrations.
\begin{lem}\label{lem:compesator_small_filtration}
Let $\mathbb{G}, \mathbb{H}$ be filtrations such that $\mathcal{G}_t
\subset \mathcal{H}_t,$ for all $t \in [0,T]$. Suppose that the
process $M_t := X_t - \int\limits_0^t {\lambda_udu}$ is a
$\mathbb{H}$-martingale, where $\lambda \ge 0$. Then the process
$M^G_t :=X_t - \int\limits_0^t {^{o}\lambda_u du}$ is a
$\mathbb{G}$-martingale, where $^{o}\lambda$ is the optional
projection of $\lambda$ onto $\mathbb{G}$.
\end{lem}
\begin{proof}
Since $\lambda_u \ge 0,$ the optional projection $^{o}\lambda$
exists and for fixed $u$, it holds that $^{o}\lambda_u =
\mathbb{E}[\lambda_u| \mathcal{G}_u]$ almost surely. If $0\le s <t$
and $H$ is bounded and $\mathcal{G}_s$-measurable, then, by Fubini's
Theorem
\begin{align*}
\mathbb{E}[H(M^G_t - M^G_s)] &= \mathbb{E}[H(X_t - X_s)] - \int\limits_s^t {\mathbb{E}[H\mathbb{E}[\lambda_u|\mathcal{G}_u]]du}\\
&= \mathbb{E}[H(X_t - X_s)] - \int\limits_s^t {\mathbb{E}[H\lambda_u]du}\\
&= \mathbb{E}[H(M_t - M_s)] = 0.\end{align*} Hence $M^G$ is a
$\mathbb{G}$-martingale.
\end{proof}

\end{document}